\documentclass[A4paper]{article}

\usepackage[normalem]{ulem}

\usepackage{amssymb, amsmath, amsthm, amsfonts, array, float, mathrsfs, quiver, verbatim}
\usepackage[labelfont=bf]{caption}
\usepackage[utf8]{inputenc}
\usepackage[hidelinks]{hyperref}
\usepackage[capitalize,noabbrev]{cleveref}
\usepackage{enumitem,hyperref}
\usepackage{authblk}
\usepackage{tikz-cd}
\usepackage{tikz}
\usetikzlibrary{arrows.meta}
\usetikzlibrary{patterns.meta}
\usetikzlibrary{calc, positioning, fit, shapes.misc}
\usetikzlibrary{shapes.geometric}

\def\rotateclockwise#1{
  \newdimen\xrw
  \pgfextractx{\xrw}{#1}
  \newdimen\yrw
  \pgfextracty{\yrw}{#1}
  \pgfpoint{\yrw}{-\xrw}
}

\def\rotatecounterclockwise#1{
  \newdimen\xrcw
  \pgfextractx{\xrcw}{#1}
  \newdimen\yrcw
  \pgfextracty{\yrcw}{#1}
  \pgfpoint{-\yrcw}{\xrcw}
}

\def\outsidespacerpgfclockwise#1#2#3{
  \pgfpointscale{#3}{
    \rotateclockwise{
      \pgfpointnormalised{
        \pgfpointdiff{#1}{#2}}}}
}

\def\outsidespacerpgfcounterclockwise#1#2#3{
  \pgfpointscale{#3}{
    \rotatecounterclockwise{
      \pgfpointnormalised{
        \pgfpointdiff{#1}{#2}}}}
}

\def\outsidepgfclockwise#1#2#3{
  \pgfpointadd{#2}{\outsidespacerpgfclockwise{#1}{#2}{#3}}
}

\def\outsidepgfcounterclockwise#1#2#3{
  \pgfpointadd{#2}{\outsidespacerpgfcounterclockwise{#1}{#2}{#3}}
}

\def\outside#1#2#3{
  ($ (#2) ! #3 ! -90 : (#1) $)
}

\def\cornerpgf#1#2#3#4{
  \pgfextra{
    \pgfmathanglebetweenpoints{#2}{\outsidepgfcounterclockwise{#1}{#2}{#4}}
    \let\anglea\pgfmathresult
    \let\startangle\pgfmathresult

    \pgfmathanglebetweenpoints{#2}{\outsidepgfclockwise{#3}{#2}{#4}}
    \pgfmathparse{\pgfmathresult - \anglea}
    \pgfmathroundto{\pgfmathresult}
    \let\arcangle\pgfmathresult
    \ifthenelse{180=\arcangle \or 180<\arcangle}{
      \pgfmathparse{-360 + \arcangle}}{
      \pgfmathparse{\arcangle}}
    \let\deltaangle\pgfmathresult

    \newdimen\x
    \pgfextractx{\x}{\outsidepgfcounterclockwise{#1}{#2}{#4}}
    \newdimen\y
    \pgfextracty{\y}{\outsidepgfcounterclockwise{#1}{#2}{#4}}
  }
  -- (\x,\y) arc [start angle=\startangle, delta angle=\deltaangle, radius=#4]
}

\def\corner#1#2#3#4{
  \cornerpgf{\pgfpointanchor{#1}{center}}{\pgfpointanchor{#2}{center}}{\pgfpointanchor{#3}{center}}{#4}
}

\def\hedgeiii#1#2#3#4{
  \outside{#1}{#2}{#4} \corner{#1}{#2}{#3}{#4} \corner{#2}{#3}{#1}{#4} \corner{#3}{#1}{#2}{#4} -- cycle
}

\def\hedgem#1#2#3#4{
  
  \outside{#1}{#2}{#4}
  \pgfextra{
    \def\hgnodea{#1}
    \def\hgnodeb{#2}
  }
  foreach \c in {#3} {
    \corner{\hgnodea}{\hgnodeb}{\c}{#4}
    \pgfextra{
      \global\let\hgnodea\hgnodeb
      \global\let\hgnodeb\c
    }
  }
  \corner{\hgnodea}{\hgnodeb}{#1}{#4}
  \corner{\hgnodeb}{#1}{#2}{#4}
  -- cycle
}

\def\hedgeii#1#2#3{
  \hedgem{#1}{#2}{}{#3}
}

\definecolor{red}{HTML}{d7191c}
\definecolor{blue}{HTML}{2c7bb6}
\definecolor{orange}{HTML}{fdae61}

\usepackage[backend=biber, maxbibnames=99, style=alphabetic]{biblatex}
\newtheorem{theorem}{Theorem}[section]

\newtheorem{proposition}[theorem]{Proposition}

\newtheorem*{thmfunctoriality-graphs}{\Cref{thm:functoriality-graphs}}
\newtheorem*{thmfunctoriality-coalgebras}{\Cref{thm:functoriality-coalgebras}}
\newtheorem*{thmfunctoriality-hypergraphs}{\Cref{thm:functoriality-hypergraphs}}

\newtheorem{lemma}[theorem]{Lemma}

\theoremstyle{definition}
\newtheorem{definition}[theorem]{Definition}

\newtheorem{example}[theorem]{Example}

\theoremstyle{remark}
\newtheorem{remark}[theorem]{Remark}

\newcommand{\xto}{\xrightarrow}

\newcommand{\power}{\mathcal{P}}
\newcommand{\Coalg}{\mathbf{Coalg}}

\newcommand{\Role}{\mathrm{Role}}
\newcommand{\Set}{\mathbf{Set}}
\newcommand{\Graph}{\mathbf{Graph}}
\newcommand{\HGraph}{\mathbf{HGraph}}
\newcommand{\FHGraph}{\mathbf{FHGraph}}

\newcommand{\define}[1]{{\bf \boldmath{#1}}}

\newcommand{\block}[2]{{#1}/{#2}}
\newcommand{\lift}[1]{\widehat{#1}}

\definecolor{NavyBlue}{rgb}{0.0, 0.0, 0.5}

\DeclareFontFamily{U}{mathb}{}
\DeclareFontShape{U}{mathb}{m}{n}{<-5.5> mathb5 <5.5-6.5> mathb6 
<6.5-7.5> mathb7 <7.5-8.5> mathb8 <8.5-9.5> mathb9 <9.5-11> mathb10 
<11-> mathb12}{}
\DeclareSymbolFont{mathb}{U}{mathb}{m}{n}
\DeclareFontSubstitution{U}{mathb}{m}{n}
\DeclareMathSymbol{\blackdiamond}{\mathbin}{mathb}{"0C} 

\crefname{item}{Condition}{Conditions}
\title{Coalgebraic analysis of social systems}

\author[1]{Nima Motamed}
\author[2]{Nina Otter}
\author[3]{Emily Roff}

\affil[1]{Utrecht University, The Netherlands}
\affil[2]{DataShape, Inria-Saclay,  France}
\affil[3]{University of Edinburgh, United Kingdom}

\begin{document}

\maketitle

\begin{abstract}
The \emph{algebraic analysis of social systems}, or algebraic social network analysis, refers to a collection of methods designed to extract information about the structure of a social system represented as a directed graph. Central among these are methods to determine the \emph{roles} that exist within a given system, and the \emph{positions}. The analysis of roles and positions is highly developed for social systems that involve only pairwise interactions among actors---however, in contemporary social network analysis it is increasingly common to use models that can take into account higher-order interactions as well. In this paper we take a category-theoretic approach to the question of how to lift role and positional analysis from graphs to hypergraphs, which can accommodate higher-order interactions. We use the framework of \emph{universal coalgebra}---a `theory of systems' with origins in computer science and logic---to formalize the main concepts of role and positional analysis and extend them to a large class of structures that includes both graphs and hypergraphs. As evidence for the validity of our definitions, we prove a very general functoriality theorem that specializes, in the case of graphs, to a folkloric observation about the compatibility of positional and role analysis.
\end{abstract}

\tableofcontents


\section{Introduction}
\label{sec:intro}

\emph{Social network analysis} studies the structure of social systems represented by combinatorial objects such as graphs. At the heart of the subject are two complementary questions: what are the \emph{positions} within a given social network, and what are the \emph{roles}? These questions were first articulated in the 1970s and early 1980s in work by Lorrain and White \cite{LorrainWhite1971}, White, Boorman and Brieger \cite{WhiteBoormanBreiger1976}, Pattison \cite{Pattison1982} and others. Over the decade that followed, a large literature accumulated around \emph{positional analysis} and \emph{role analysis}, each of which comes with its own conceptual and computational challenges; a comprehensive account is given by Wasserman and Faust in \cite[Part IV]{Wasserman_Faust_1994}. The methods developed during that era are still in use---indeed, there is increasingly sophisticated software dedicated to their implementation \cite{Ostoic2020}. From today's perspective, however, they have an important limitation: they are restricted to systems modelled by directed graphs, which can only represent \emph{pairwise} interactions among actors.

Modern network science demands models that can capture \emph{higher‑order} relations: patterns of interaction among groups of actors, not just pairs \cite{BGHS}. Graphs cannot represent such relations; for that, one needs a richer combinatorial structure, such as a \emph{hypergraph}. As things stand, the study of roles and positions does not extend from graphs to hypergraphs. The primary challenge is to find the right definitions: how should we define a `role' with respect to a higher-order relation, and how should we define a `position'? A robust answer to these questions calls for a framework that encompasses both graphs and hypergraphs, and in which both roles and positions have natural interpretations. This paper provides that framework. 

To make the paper accessible to a mixed readership of mathematicians and computational social scientists, we open with an extended introduction that summarizes the main ideas of `traditional' role and positional analysis before explaining our framework and how it offers an answer to the central questions. The introduction closes with a Reader's Guide to the rest of the paper which should help those unfamiliar with category theory to navigate it.


\paragraph{Role analysis}

\begin{figure}
    \centering
    \begin{tikzpicture}[line cap=round, line join=round, thick, yscale=0.8] 
    \tikzset{ dot/.style={circle, fill, inner sep=1.6pt}, 
    over/.style={preaction={draw=white, line width=\pgflinewidth+1.2pt}},}
    \def\bend{16}
    \def\gap{2.5pt} 

    \node[dot] (a) at (-2.1, 2.2) {}; 
    \node[dot] (b) at (-0.7, 2.2) {}; 
    \node[dot] (c) at (-2.8, 0.7) {}; 
    \node[dot] (d) at (-1.4, 0.7) {}; 
    \node[dot] (e) at (-0, 0.7) {}; 
    \node[dot] (f) at (-0.7, -0.8) {}; 


    \draw[->, shorten <=\gap, shorten >=\gap, black] (a) -- (d); 
    \draw[->, shorten <=\gap, shorten >=\gap, black] (b) -- (d); 
    \draw[->, shorten <=\gap, shorten >=\gap, black] (d) -- (f);
    \draw[->, shorten <=\gap, shorten >=\gap, black] (e) -- (f);
    \draw[->, shorten <=\gap, shorten >=\gap, black] (c) -- (f);

    \draw[->, shorten <=\gap, shorten >=\gap, red, dotted] (a) to[bend left=\bend] (b); 
    \draw[->, shorten <=\gap, shorten >=\gap, blue, dashed] (b) to[bend left=\bend] (a);
    \draw[->, shorten <=\gap, shorten >=\gap, blue, dashed] (c) to[bend left=\bend] (d); 
    \draw[->, shorten <=\gap, shorten >=\gap, blue, dashed] (d) to[bend left=\bend] (c);
    \draw[->, shorten <=\gap, shorten >=\gap, blue, dashed] (e) to[bend left=\bend] (d); 
    \draw[->, shorten <=\gap, shorten >=\gap, red, dotted] (d) to[bend left=\bend] (e);
    
    \draw[dashed] (-3.2,-1.3) rectangle (0.4,2.7); 
    \end{tikzpicture}
    \quad
    \begin{tikzpicture}[line cap=round, line join=round, thick, yscale=0.8] 
    \tikzset{ dot/.style={circle, fill, inner sep=1.6pt}, 
    over/.style={preaction={draw=white, line width=\pgflinewidth+1.2pt}},}
    \def\bend{16}
    \def\gap{2.5pt} 

    \node[dot] (a) at (-2.1, 2.2) {}; 
    \node[dot] (b) at (-0.7, 2.2) {}; 
    \node[dot] (c) at (-2.8, 0.7) {}; 
    \node[dot] (d) at (-1.4, 0.7) {}; 
    \node[dot] (e) at (-0, 0.7) {}; 
    \node[dot] (f) at (-0.7, -0.8) {}; 


    \draw[->, shorten <=\gap, shorten >=\gap, black] (a) -- (d); 
    \draw[->, shorten <=\gap, shorten >=\gap, black] (b) -- (d); 
    \draw[->, shorten <=\gap, shorten >=\gap, black] (d) -- (f);
    \draw[->, shorten <=\gap, shorten >=\gap, black] (e) -- (f);
    \draw[->, shorten <=\gap, shorten >=\gap, black] (c) -- (f);

    \draw[->, shorten <=\gap, shorten >=\gap, purple, dashdotted] (a) to[bend left=\bend] (b); 
    \draw[->, shorten <=\gap, shorten >=\gap, purple, dashdotted] (b) to[bend left=\bend] (a);
    \draw[->, shorten <=\gap, shorten >=\gap, purple, dashdotted] (c) to[bend left=\bend] (d); 
    \draw[->, shorten <=\gap, shorten >=\gap, purple, dashdotted] (d) to[bend left=\bend] (c);
    \draw[->, shorten <=\gap, shorten >=\gap, purple, dashdotted] (e) to[bend left=\bend] (d); 
    \draw[->, shorten <=\gap, shorten >=\gap, purple, dashdotted] (d) to[bend left=\bend] (e);
    
    \draw[dashed] (-3.2,-1.3) rectangle (0.4,2.7); 
    \end{tikzpicture}
    \quad
    \begin{tikzpicture}[line cap=round, line join=round, thick, yscale=0.8] 
    \tikzset{ dot/.style={circle, fill, inner sep=1.6pt}, } 
    \def\bend{22}
    \def\gap{2.5pt} 
    
    \node[dot] (a) at (-1.6, 2.2) {}; 
    \node[dot] (b) at (-1.2, 2.2) {}; 
    \node[dot] (d) at (-1.4, 0.8) {}; 
    \node[dot] (d2) at (-1.8, 0.8) {}; 
    \node[dot] (d3) at (-1, 0.8) {}; 
    \node[dot] (e) at (-1.4, -0.45) {}; 
    \node[rectangle,draw,fit=(a) (b),inner sep=.5] (ab) {};
    \node[rectangle,draw,fit=(d) (d2) (d3),inner sep=.5] (dd) {};
    \node[rectangle,draw,fit=(e),inner sep=.5] (ee) {};
    
    \draw[->, shorten <=\gap, shorten >=\gap, black] (ab) -- (dd); 
    \draw[->, shorten <=\gap, shorten >=\gap, black] (dd) -- (ee); 
    \path[->] (ab) edge [shorten <=\gap, shorten >=\gap, out=45, in=315, looseness=4, loop, purple, dashdotted] (ab);
    \path[->] (dd) edge [shorten <=\gap, shorten >=\gap, out=45, in=315, looseness=4, loop, purple, dashdotted] (dd);
    \path[->] (ee) edge [shorten <=\gap, shorten >=\gap, out=45, in=315, looseness=4, loop, purple, dashdotted] (ee);
        
    \draw[dashed] (-3.2,-1.3) rectangle (0.4,2.7); 
    \end{tikzpicture}
    
    \caption{On the left, a multirelational network containing a `parent' relation \(P\) (solid arrows), `sister' relation \(S\) (dotted arrows), and `brother' relation \(B\) (dashed arrows). Center, a role reduction combining \(S\) and \(B\) into the `sibling' relation \(\overline{S}\) (dash-dotted arrows). On the right, a blockmodel by `generations'.}
    \label{fig:family-network}
\end{figure}
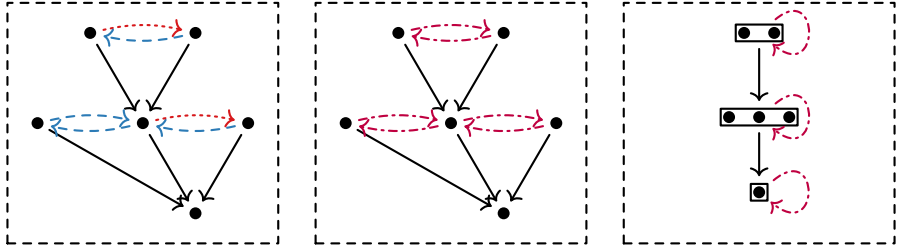

In traditional social network analysis, the basic object of interest is a \emph{multirelational network}: a set \(A\) of social actors and a handful of binary relations \(R_1,\ldots, R_k\) among the actors in \(A\), which we will refer to in this introduction as the \emph{generating relations}. In running examples throughout the paper, we will consider a multirelational network \(F = (A, (S,B,P))\) in which \(A\) comprises the members of a family and the generating relations are the `sister' relation \(S\), the `brother' relation \(B\), and the `parent' relation \(P\). \Cref{fig:family-network} shows how this network can be depicted as an edge-coloured directed graph.\footnote{For consistency and to aid understanding, almost all examples in this paper will be of kinship systems. However, the types of systems that are studied in social network analysis are manifold, and include: advice, mentorship and coworking relationships between employees in a firm \cite{lazega}; sympathy/antipathy between monks in a monastery \cite{sampson, BBA75};  marriage and business partnerships between families in  Fifteenth Century Florence \cite{ BreigerPattison1986}; favour exchanges between villages in rural India \cite{villages}; relationship status---couples, exes, `it's complicated', friends, enemies, and family---between penguins in Kyoto Aquarium \cite{Penguins_of_Kyoto_dataset}.
}

The \emph{roles} in a multirelational network are relations among the actors that can be obtained by composing the generating relations. In our family network, the roles include `sister', `brother' and `parent', but also `brother-of-parent' (commonly known as `uncle'), `sister-of-parent-of-parent' (`great-aunt'), and so forth. The set of all roles in a multirelational network \(G\) forms a semigroup, \(\Role(G)\): the \emph{semigroup of roles} in \(G\).

Even for quite a small network, the semigroup of roles can easily be too large for immediate interpretation. One way to condense the information it carries is to look for quotients: smaller semigroups whose structure reflects that of \(\Role(G)\). We will call these \emph{role reductions}. For instance, in the case of the family network \(F\), there is a role reduction in which the `brother' and `sister' roles are combined into a role that one might call `sibling' (see \Cref{fig:parent-sibling} and \Cref{eg:parent-sibling}). Because it is a semigroup quotient---meaning that it respects the composition of relations---this reduction also has the effect of combining the role of `uncle' with `aunt', `great-aunt' with `great-uncle', and so forth. The study of \(\Role(G)\) and its role reductions is known as \emph{role analysis}.


\paragraph{Positional analysis}

Whereas a role is an equivalence class of relations among actors, a \emph{position} in a network is an equivalence class of actors. \emph{Positional analysis} refers to the process of aggregating actors into positions. The output of positional analysis is a \emph{blockmodel}: a smaller network in which the nodes are positions and the relations reflect the structure of the original network. Formally, a blockmodel of a multirelational network \(G = (A, (R_1,\ldots, R_k))\) is a simultaneous quotient of the graphs \((A,R_1), \ldots, (A,R_k)\).

The idea is that a position should be comprised of actors who are equivalently embedded in the network; the difficulty is that there are many different ways one might specify what it means to be ``equivalently embedded''. A good notion of equivalence ought to be relaxed enough to produce \emph{interesting} blockmodels, yet strict enough that the blockmodels it produces are \emph{sensible}. 

One way to judge whether a blockmodel of \(G\) is sensible is to ask whether it induces a role reduction---a quotient of \(\Role(G)\). For instance, there is a blockmodel of our family network \(F\) in which the positions are the `generations' in the family. This blockmodel does induce a role reduction: it has the effect of combining `parent' and `parent-of-sibling' into a single role (see \Cref{fig:generations} and \Cref{eg:generations}). But that is not always the case; indeed, a poor choice of blockmodel may \emph{increase} the number of relations in the semigroup of roles.

To guarantee compatibility with role analysis, early approaches to positional analysis demanded that actors occupying the same position should be \emph{structurally equivalent} (meaning they have identical neighbourhoods within the network) or \emph{automorphically equivalent} (meaning there is a graph automorphism that interchanges them). These notions of equivalence are sensible---the blockmodel by a structural or automorphic equivalence always induces a role reduction---but, as Pattison explains in \cite[p.92]{Pattison1982}, too restrictive to capture ``the notion  of identical `abstract', rather than `concrete', social positions''. For instance, she writes, ``the leaders of two identically structured but distinct small groups'' ought to be regarded as occupying the same social position, but are not structurally equivalent.

A more flexible approach, introduced in 1978 by Sailer \cite{Sailer1978}, asks that actors in the same position be \emph{regularly equivalent}: that they relate in matching ways to \emph{equivalent} others. (For the formal definition, see \Cref{def:reg_equiv}.) This notion of equivalence achieves the necessary balance: strict enough to be compatible with role analysis, yet relaxed enough to capture Pattison's ``abstract positions''. In this paper, a blockmodel of a multirelational network by a regular equivalence will be called a \emph{positional reduction}.

Following Otter and Porter \cite{OtterPorter2020}, we frame the compatibility of role and positional analysis in terms of a \emph{functoriality theorem}. The theorem, which has the status of folklore in the social network literature, says that every positional reduction induces a role reduction, and it also says a little more: that if one performs a sequence of positional reductions, the induced role reductions behave consistently. (For the reader who is meeting categories and functors for the first time, we recommend Leinster's \cite[Chapter 1]{Leinster_2014} as a primer.)

\begin{thmfunctoriality-graphs}[\cite{Pattison1982, OtterPorter2020}]
    Let \(k\Graph_{\mathsf{PR}}\) denote the category of \(k\)-relational networks and positional reductions, and \(\mathbf{SGrp}_\mathsf{Q}\) the category of semigroups and their quotient maps. The semigroup of roles defines a functor
    \[\Role\colon k\Graph_{\mathsf{PR}} \to \mathbf{SGrp}_\mathsf{Q}.\]
    In particular, every positional reduction induces a role reduction.
\end{thmfunctoriality-graphs}

In this paper, we propose definitions for \emph{roles} and \emph{positions} in social systems involving higher-order relations. As evidence that our definitions are sound, we will prove that positional analysis for higher-order relations is compatible with role analysis: they satisfy the natural analogue of \Cref{thm:functoriality-graphs}.


\paragraph{A coalgebraic framework}

Our starting point is to view a social system not as something static but as a type of \emph{transition system}, and to use \emph{universal coalgebra} to set up a common framework for positional and role analysis. Universal coalgebra is the modern mathematics of transition systems; we recommend Jacobs \cite{Jacobs2017} as an introduction to the field.

In its simplest form, a coalgebra is a set \(A\) equipped with a function \(A \to T(A)\) where \(T(A)\) is some other set associated in a systematic way to \(A\). (Formally, \(T\) should be an endofunctor on the category \(\Set\) of sets and functions.) In the computer science context from which the study of coalgebra emerged, \(A\) is usually thought of as a set of states and \(T(A)\) as the set of possible outcomes of some process of computation. For us, \(A\) will instead represent the set of actors in a social system, and the function \(A \to T(A)\) will assign to each actor their `neighbourhood' in a specified sense.

As we explain in \Cref{sec:ss_coalg}, both graphs and hypergraphs can be encoded as coalgebras by taking \(T\) to be, respectively, the covariant powerset functor \(\power\) or its square \(\power\power\) (defined in \Cref{def:powerset}). That fact is well-known; what makes it important to our story is the observation that a regular equivalence on a graph is precisely a \emph{bisimulation equivalence} on the corresponding \(\power\)-coalgebra. Just as regular equivalences are at the heart of social network analysis, so bisimulation equivalences are at the heart of the study of transition systems, and of coalgebra. 

Indeed, the appearance of this notion in the social science literature of the 1970s is rather remarkable. It is a striking and well-known fact about bisimulation that it emerged independently in three fields at once: Sangiorgi has traced the parallel development of the concept in computer science, philosophical logic and mathematics through the 1970s \cite{Sangiorgi2009}. According to Sangiorgi's account, the term `simulation' first appears in this context in the work of Milner in 1970 and 1971, and `bisimulation' eventually emerges in a 1980 collaboration between Milner and Park. It is remarkable, then, that social scientists converged upon the same concept over almost exactly the same period---and, so far as we can tell, also independently. We regard this as a powerful omen that universal coalgebra---and category theory, more broadly---offers the right apparatus by which to lift techniques of social network analysis from graphs to hypergraphs.

Our category-theoretic framework is constructed in Sections \ref{sec:positional_bisim} and \ref{sec:role_kleisli}. \Cref{sec:positional_bisim} sets up positional analysis for the coalgebras for an arbitrary functor \(T \colon \Set \to \Set\). A \emph{blockmodel} for a \(T\)-coalgebra is a quotient by a \(T\)‑bisimulation equivalence; a \emph{positional reduction} is a coequalizer in the category \(k\Coalg(T)\) of \emph{\(k\)-relational \(T\)-coalgebras} (\Cref{def:multirelational-coalgebras}).

In \Cref{sec:role_kleisli} we turn to role analysis. Here, the key is that `traditional' role analysis for graphs can be interpreted in terms of composition in the \emph{Kleisli category} \(\Set_\mathbb{P}\) of the \emph{powerset monad} \(\mathbb{P} = (\power, \mu, \eta)\): for a multirelational network \(G = (A, (R_i)_{i=1}^k)\), the semigroup \(\Role(G)\) is a sub-semigroup of the monoid of endomorphisms \(\Set_\mathbb{P}(A,A)\). (These terms will be defined in \Cref{subsec:k-coalgebra-roles}.) Since we are only interested in its semigroup structure, it is not actually important that \(\Set_\mathbb{P}(A,A)\) forms a monoid, nor that \(\power\) carries the full structure of a monad; what matters is that it carries the structure of a \emph{semimonad}. For any semimonad \(\mathbb{T} = (T, \mu)\) on \(\Set\), we define a \emph{semigroup of \(\mu\)-roles} and a notion of \emph{\(\mu\)-role reduction} for \(k\)-relational \(T\)-coalgebras. 

Our main theorem, which we present as a warrant of the correctness of our definitions, says that \(\mu\)-role analysis is functorial with respect to \(T\)-positional reductions. Thus, role and positional analysis are unified within a category-theoretic framework that subsumes the case of graphs.

\begin{thmfunctoriality-coalgebras}
    Let \(\mathbb{T} = (T, \mu)\) be a semimonad on \(\Set\), and let \(k\Coalg(T)_\mathsf{PR}\) denote the category of \(k\)-relational \(T\)-coalgebras and positional reductions. The semigroup of roles defines a functor
    \[\Role_\mathbb{T} \colon k\Coalg(T)_\mathsf{PR} \to \mathbf{SGrp}_\mathsf{Q}.\]
    In particular, every \(T\)-positional reduction induces a \(\mu\)-role reduction.
\end{thmfunctoriality-coalgebras}


\paragraph{Extending positional and role analysis} 

In the case of \(\power\)-coalgebras, \Cref{thm:functoriality-coalgebras} recovers \Cref{thm:functoriality-graphs}, expressing the compatibility of `traditional' role and positional analysis. But \Cref{thm:functoriality-coalgebras} is vastly more general, and its value lies in its power to indicate \emph{new} methods of role and positional analysis that are better suited to the complexities of modern social scientific data. Developing such methods is a challenge of active interest to researchers in the field; in recent years, for instance, roles and positions have been extended in more than one way to fuzzy social networks \cite{FanLiauLin2007, MicicStanimirovicJanvic2023} and to weighted or valued networks \cite{Nordlund2020}.

The coalgebraic framework suggests a `universal' approach to this challenge. Any type of social system that can be modelled by the coalgebras for a functor \(T \colon \Set \to \Set\) admits a natural notion of regular equivalence and of positional reductions. If \(T\) happens to carry the structure of a semimonad, then there is a corresponding notion of role analysis, and \Cref{thm:functoriality-coalgebras} guarantees that role analysis will be functorial with respect to positional reductions.

In \Cref{sec:analysis-for-hypergraphs} we implement this strategy in the case of the functor \(\power\power \colon \Set \to \Set\) that takes a set \(A\) of actors to the set \(\power\power(A)\) of sets of sets of actors in \(A\). Coalgebras for \(\power\power\) correspond exactly to a certain class of directed hypergraphs known as F-hypergraphs, and naturally encompass all undirected hypergraphs. Interpreting bisimulation equivalences for  \(\power\power\)-coalgebras thus yields a natural notion of \emph{regular equivalence} and \emph{positional reduction} for hypergraphs.

Crucially, \(\power\power\) admits a semimonad structure; in fact it admits two, giving rise to two associative composition operations for \(\power\power\)-coalgebras (\Cref{eg:double-powerset-roles}). Thus, our framework produces not just one, but two distinct semigroups of roles in a multirelational hypergraph: a semigroup of \emph{tight roles}, which keeps track of `branch' structure in higher‑order interactions, and a semigroup of \emph{loose roles}, which does not. By applying \Cref{thm:functoriality-coalgebras}, we prove that both tight and loose role analysis are compatible with positional analysis for hypergraphs.

\begin{thmfunctoriality-hypergraphs}
    The semigroups of tight and loose roles are both functorial with respect to positional reductions of hypergraphs.
\end{thmfunctoriality-hypergraphs}


\paragraph{The structure of the paper}

We begin in \Cref{sec:positional_role} by giving a mathematical account of the main ideas of positional and role analysis on graphs. In \Cref{sec:ss_coalg} we motivate the use of hypergraph models through the example of a co-authorship network, then introduce coalgebras and explain how they can encode both graphs and hypergraphs. \Cref{sec:positional_bisim} formulates positional analysis for coalgebras via bisimulation. \Cref{sec:role_kleisli} formulates role analysis via Kleisli composition and proves our key result, \Cref{thm:functoriality-coalgebras}, which expresses the compatibility of role analysis with positional analysis in the coalgebraic framework.

In \Cref{sec:analysis-for-hypergraphs} we specialize to hypergraphs, giving explicit definitions of regular equivalences and positional reductions in that setting, then introducing the semigroups of tight and loose roles and applying \Cref{thm:functoriality-coalgebras} to establish their functoriality with respect to positional reductions. In \Cref{sec:future_work} we conclude with ideas to extend our framework---for instance, to a wider class of directed hypergraphs, or to accommodate notions of `approximate equivalence'.

\paragraph{A reader's guide} 

We have endeavoured to write this paper for distinct and probably almost disjoint sets of readers: those interested in the practical applications, who may have no knowledge of category theory, and those who are more categorically-minded but may be completely unfamiliar with the applications. For the reader who is interested \emph{only} in the question of how to study roles and positions in graphs and hypergraphs, this Introduction, together with Section \ref{sec:positional_role}, Section \ref{subsec:higher-order-relations}, and Sections \ref{subsec:hypergraph-positional-analysis}-\ref{subsec:hypergraph-role-analysis}, should provide a self-contained account of the relevant definitions. We hope that readers with a basic knowledge of category theory---say, to the level of Leinster's textbook \cite{Leinster_2014}---will find the entire paper accessible and interesting.


\section{Positional and role analysis on graphs}
\label{sec:positional_role}

Both positional analysis and role analysis are concerned with reducing a model of some social system to one that is smaller and easier to comprehend. In this section we give an overview of both methodologies, in mathematical terms that are suited to the later developments in this paper. The interested reader can find full accounts of these methodologies in the classical monographs \cite{Pattison1993, Wasserman_Faust_1994} and more modern accounts of some aspects in \cite{rawlings_network_2023,DoreianBatageljFerligoj2004, brandes2005network, Pattison2023}.


\subsection{Positional analysis}\label{subsec:positional}

In the field of social network analysis, a \emph{position} within a network is understood to be a set of nodes that share a similar pattern of ties to other nodes. The aim of \emph{positional analysis} is to partition the set of nodes in a network into positions, and assign ties among the positions that reflect the structure of the original network. In this section we explain how that can be made precise for networks modelled by directed graphs. For brevity, we refer to directed graphs as \emph{graphs}.

\begin{definition}\label{def:graph}
    A \define{graph} is a pair of sets \(G = (A, R)\), where \(R \subseteq A \times A\). The elements of \(A\) are called the \define{vertices} of \(G\) and the pairs \((v,w) \in R\) are called \define{edges}; in diagrams, we depict an edge \((v,w)\) as an arrow \(v \to w\).
    
    A \define{map of graphs} $f \colon (A, R) \to (A', R')$ is a function $f\colon A \to A'$ that preserves edges: for every $(a, b) \in R$ we have $(f(a), f(b)) \in R'$. We will denote the category of graphs and maps of graphs by \(\Graph\).
\end{definition}

Given a social network modelled by a graph \(G\), the aim of positional analysis is to construct a smaller graph whose vertices deserve to be understood as positions in the network. We will follow the social science literature in referring to this smaller graph as a \emph{blockmodel}. This term appears first in work by White, Boorman and Breiger \cite[pp.~739-740]{WhiteBoormanBreiger1976}; a more formal definition is given by Arabie, Boorman and Levitt \cite[Definition 1]{ArabiePhippsLevitt1978}. White and Reitz refer to blockmodels as \emph{full homomorphic images} \cite[Definition 5]{WhiteReitz1983}. Mathematicians would call a blockmodel a \emph{quotient} of \(G\).

\begin{definition}\label{def:blockmodel}
    Let \(G = (A, R)\) be a graph, and let \(E \subseteq A \times A\) be an equivalence relation. The \define{blockmodel} of \(G\) by \(E\) is the graph \(\block{G}{E} = \left(\block{A}{E}, \block{R}{E}\right)\) whose vertex set \(\block{A}{E}\) consists of the equivalence classes for \(E\) and whose edge set is
    \[\block{R}{E} = \left\{(X,Y) \mid \text{there exist \(x \in X\) and \(y \in Y\) such that \((x,y) \in R\)}\right\}.\]
    We will sometimes refer to the elements of \(\block{A}{E}\) as \define{blocks} or as \define{positions}.
\end{definition}

\begin{figure}
    \centering
    \begin{tikzpicture}[line cap=round, line join=round, thick, yscale=0.8] 
    \tikzset{ dot/.style={circle, fill, inner sep=1.6pt}, } 
    \def\bend{16}
    \def\gap{2.2pt} 

    \node[anchor=base west] (a) at ( 0, 0) {$a$}; 
    \node[anchor=base west] (b) at (-.5, 2) {$b$\strut}; 
    \node[anchor=base west] (c) at ( .5, 2) {$c$\strut}; 

    \draw[->, shorten <=-\gap, shorten >=\gap] (b) -- (a); 
    \draw[->, shorten <=-\gap, shorten >=\gap] (c) -- (a); 

    \node[anchor=north] (d) at (.25, -1) {$G$\strut}; 

    \draw[dashed] (-1,-.7) rectangle (1.5,3); 
    \end{tikzpicture}
    \qquad
    \begin{tikzpicture}
    [line cap=round, line join=round, thick, yscale=0.8] 
    \tikzset{ dot/.style={circle, fill, inner sep=1.6pt}, } 
    \def\bend{16}
    \def\gap{2.2pt} 

    \node[anchor=base west] (a) at ( 0, 0) {$a$};
    \node[anchor=base west] (b) at (-.25, 0) {$b$\strut}; 
    \node[anchor=base west] (c) at ( .5, 2) {$c$\strut}; 
    \node[rectangle,draw,fit=(a) (b),inner sep=.5] (ba) {};
    \node[rectangle,draw,fit=(c),inner sep=.5] (cc) {};

    \node[anchor=north] (d) at (.25, -1) {$\block{G}{E_1}$\strut}; 

    \draw[->, shorten <=\gap, shorten >=\gap] (cc) -- (ba); 

    \draw[->] (-.35,.3) to [out=140, in=220, loop, looseness=4] (-.35,-.1);
        
    \draw[dashed] (-1,-.7) rectangle (1.5,3); 
    \end{tikzpicture}
    \qquad
    \begin{tikzpicture}[line cap=round, line join=round, thick, yscale=0.8] 
    \tikzset{ dot/.style={circle, fill, inner sep=1.6pt}, } 
    \def\bend{16}
    \def\gap{2.2pt} 

    \node[anchor=base west] (a) at ( 0, 0) {$a$\strut}; 
    \node[anchor=base west] (b) at (-.1, 2) {$b$\strut}; 
    \node[anchor=base west] (c) at ( .1, 2) {$c$\strut}; 

    \node[rectangle,draw,fit=(a),inner sep=.5] (aa) {};
    \node[rectangle,draw,fit=(b) (c),inner sep=.5] (bcd) {};

    \node[anchor=north] (d) at (.25, -1) {$\block{G}{E_2}$\strut}; 
    
    \draw[->, shorten <=\gap, shorten >=\gap] (bcd) -- (aa); 
    
    \draw[dashed] (-1,-.7) rectangle (1.5,3); 
    \end{tikzpicture}
    
    \caption{A graph \(G\) and its blockmodels with respect to two equivalence relations, of which only \(E_2\) is a regular equivalence. See \Cref{eg:graph-blockmodels}.}
    \label{fig:graph-blockmodels}
\end{figure}
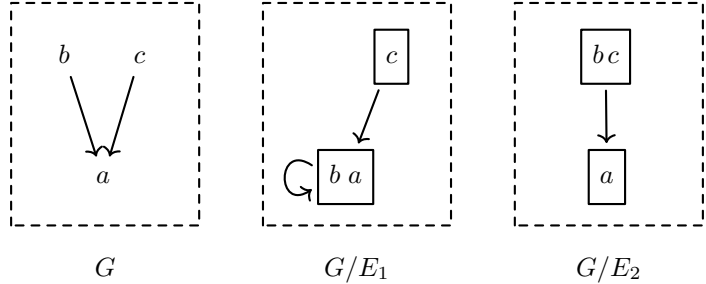

Most blockmodels do not provide good descriptions of the positions within a network. For instance, every graph \(G\) with nonempty vertex set admits as a blockmodel a graph with a single vertex (and a single loop, if the edge-set of \(G\) is also nonempty)---but it is seldom useful to think of a network as containing a single position. In the early literature on social network analysis, a central objective was to specify classes of equivalence relations on graphs that seemed to capture the notion of `equivalent positions', relating vertices only if they share the same pattern of connection with other vertices. Since there are various ways to interpret what a `pattern of connection' should be, there are various candidate classes of equivalence relations, including the \emph{structural equivalences} and \emph{automorphic equivalences} mentioned in the Introduction. More general than either of these is the class of \emph{regular equivalences}.

\begin{definition}\label{def:reg_equiv}
Let $G = (A, R)$ be a graph, \(E\) an equivalence relation on \(A\).
\begin{enumerate}
    \item We will say that \(E\) is \define{outward-regular} if, for every $(a, a') \in E$ and for each $(a, b) \in R$, there is $(a', b') \in R$ such that $(b, b') \in E$. \label{cond:out-reg_equiv}

    \item We will say that \(E\) is \define{inward-regular} if, for every $(a, a') \in E$ and for each $(b, a) \in R$, there is $(b', a') \in R$ such that $(b, b') \in E$. \label{cond:in-reg_equiv}
    \end{enumerate}
A \define{regular equivalence} on \(G\) is an equivalence relation \(E \subseteq A \times A\) which is both outward-and inward-regular.
\end{definition}

The term \emph{regular equivalence} is due to White and Reitz \cite[Definition 11]{WhiteReitz1983}. The terms \emph{outward-regular} and \emph{inward-regular} are non-standard.

Condition \ref{cond:out-reg_equiv} in \Cref{def:reg_equiv} obliges an equivalence relation \(E\) to respect the `pattern of outgoing connections' of each vertex: it says that if \(E\) relates \(v\) and \(v'\), then these two vertices must admit outgoing edges to some vertices \(w\) and \(w'\) that are themselves related by \(E\). Condition \ref{cond:in-reg_equiv} obliges \(E\) to respect the `pattern of incoming connections' in the same way. Consequently, blockmodels by outward- or inward-regular equivalences offer more faithful representations of the positions in a network and the ties between them.

\begin{example}\label{eg:graph-blockmodels}
    Consider the graph \(G = (A,R)\) shown on the left in \Cref{fig:graph-blockmodels}. For the sake of illustration, suppose that \(G\) represents a small family unit, and \(R\) is the `parent' relation: \((v,v') \in R\) if \(v'\) is the parent of \(v\). \Cref{fig:graph-blockmodels} also shows blockmodels of \(G\) with respect to two different equivalence relations on \(A\).
    
    The equivalence relation \(E_1\) is not inward-regular: it treats \(b\) as equivalent to \(a\) although \(a\) admits an incoming edge (\(a\) `is a parent') while \(b\) does not. It is not outward-regular either, since it treats \(a\) as equivalent to \(b\) although \(b\) admits an outgoing edge (\(b\) `has a parent') while \(a\) does not.

    The equivalence relation \(E_2\), however, is a regular equivalence. The positions in \(\block{G}{E_2}\) correspond to the `generations' of this small family.
\end{example}

Of the three notions introduced in \Cref{def:reg_equiv}, regular equivalence, which respects both outgoing and incoming connections, arguably captures the idea of `equivalent positions' most accurately. Yet, there are good reasons to consider the larger classes of outward-~or inward-regular equivalences. Pattison \cite[p.389]{Pattison1988} argues for their usefulness from a social-scientific perspective, and there are also mathematical reasons to allow this generality. As we will see in \Cref{subsec:graph-functoriality}, either outward- or inward-regularity is enough on its own to guarantee the compatibility of positional analysis with \emph{role analysis}, described in \Cref{subsec:role}. Later, in \Cref{sec:positional_bisim}, we will find that outward- and inward-regular equivalences are especially natural to consider from the perspective of universal coalgebra.


\subsection{Role analysis}\label{subsec:role}

Often one is interested in more than one type of relation among the same set of actors---that is, more than one graph structure on the same set of vertices. Role analysis concerns the structure of such models, which we will refer to as \emph{multirelational graphs}. (Elsewhere in the literature, different terms are used; for example, \cite{Ostoic2020} and \cite{PioLopezEtAl2021} refer to these as `multiplex graphs'.)

\begin{definition}\label{def:multirelational-graph}
     A \define{\(k\)-relational graph} is a pair \(G = (A, (R_i)_{i=1}^k)\) in which each \(G_i = (A,R_i)\) is a graph. A \define{map of \(k\)-relational graphs}
     \[(A, (R_i)_{i=1}^k) \to (A', (R_i')_{i=1}^k)\]
     is a function \(A \to A'\) which is a map of graphs \((A, R_i) \to (A', R_i')\) for each \(i \in \{1,\ldots,k\}\).
     We denote the category of \(k\)-relational graphs and these maps by $k\Graph$. We will also refer to \(k\)-relational graphs as \define{multirelational graphs}.
\end{definition}

Role analysis studies the \emph{compound} or \emph{composite} relations generated by a multirelational graph. Canonical early papers on this topic include Boorman and White \cite{BoormanWhite1979} and Pattison \cite{Pattison1982}.

Consider, for example, the multirelational graph shown in \Cref{fig:family-network}. Suppose the vertices are members of a family and the relation \(P\) is the `parent' relation while \(S\) is the `sister' relation and \(B\) is the `brother' relation. An anthropologist studying European family structures will no doubt be aware of the so-called `aunt' relation, which can be described as the \emph{composite} of the parent and the sister relation, in the following sense.

\begin{definition}\label{def:graph-composition}
    Given two graphs \(G = (A, R_1)\) and \(H = (A, R_2)\) on the same set of vertices, their \define{composite} is the graph \(H \ast G = (A, R_2 \ast R_1)\) where
    \[R_2 \ast R_1 = \{(v,w) \mid \text{there exists } u \in A \text{ such that } (u, w) \in R_2 \text{ and } (v,u) \in R_1\}.\]
    We will sometimes denote \(R_2 \ast R_1\) by \(R_2R_1\).
\end{definition}

Knowing the significance of the `aunt' relation, our anthropologist may find herself wondering what other composite relations exist within this social system. To answer this question amounts to considering what is called the \emph{semigroup of roles} in \(G\). To define the semigroup of roles, let us first observe that the binary operation \(\ast\) is associative, so the set of all graph structures on \(A\) forms a semigroup under composition. We will denote that semigroup by \(\mathbb{G}(A)\).

\begin{definition}
    The \define{semigroup of roles} \(\Role(G)\) in a multirelational graph \(G = (A, (R_i)_{i=1}^k)\) is the sub-semigroup of \(\mathbb{G}(A)\) generated by \((R_i)_{i=1}^k\). We write
    \[\Role(G) = \langle (A,R_1), \ldots, (A,R_k) \rangle_\ast.\]
\end{definition}

\begin{remark}\label{rem:GA-monoid}
    In fact, \(\mathbb{G}(A)\) is a monoid: its unit element is the graph whose edge set is \(\delta_A = \left\{(a,a) \mid a \in A \right\}\). Typically, though, \(\Role(G)\) will not be a submonoid, because it need not contain \(\delta_A\). The monoid \(\mathbb{G}(A)\) also has an absorbing element: the empty relation \(\emptyset \subseteq A \times A\) satisfies \(\emptyset \ast R = \emptyset = R \ast \emptyset\) for every \(R \in \mathbb{G}(A)\). We will denote the empty relation by \(0\).
\end{remark}

The semigroup of roles can easily be much larger than the generating set of relations---indeed, it quickly becomes too large to be readily comprehensible. The goal of role analysis is to identify a smaller semigroup whose structure still captures information about the structure of \(\Role(G)\). Formally, just as in positional analysis one is looking for a quotient of the graph modelling a social system, in role analysis one is looking for a quotient of the semigroup of roles. We will call such a quotient a \emph{role reduction}.

\begin{definition}\label{def:role-reduction}
    A \define{role reduction} of a multirelational graph \(G\) is a semigroup \(\mathbb{S}\) equipped with a surjective semigroup homomorphism \(\Role(G) \twoheadrightarrow \mathbb{S}\).
\end{definition}

\begin{example}\label{eg:parent-sibling}
    Consider the multirelational graph \(F = (A, (S,B,P))\) at the upper left of \Cref{fig:parent-sibling}; it is a fragment of the family in \Cref{fig:family-network}. At the upper right of \Cref{fig:parent-sibling} is the multiplication table of \(\Role(F)\). 
    
    Inspecting the table, one can see there is a congruence on \(\Role(F)\) in which \(S\) is equivalent to \(B\); the multiplication table of the quotient semigroup \(\mathbb{S}\) is shown at the lower right of \Cref{fig:parent-sibling}. The semigroup \(\mathbb{S}\) is generated by the elements \(\overline{S}\) and \({P}\), so we have \(\mathbb{S} = \Role(\overline{F})\) where \(\overline{F} = (A, (\overline{S}, {P}))\) is the multirelational graph shown at the lower left of the figure.
    
    This role reduction can be interpreted as combining the distinct roles of `sister' and `brother' into the single role of `sibling'. By doing so, it reduces the semigroup of roles from seven non-zero elements to four. For the larger family in \Cref{fig:family-network}, this reduction would also have the effect of combining the roles of `aunt' (\(S \ast P\)) and `uncle' (\(B \ast P\)) into the single role of `parent's sibling' (\(\overline{S} \ast P\)).
\end{example}

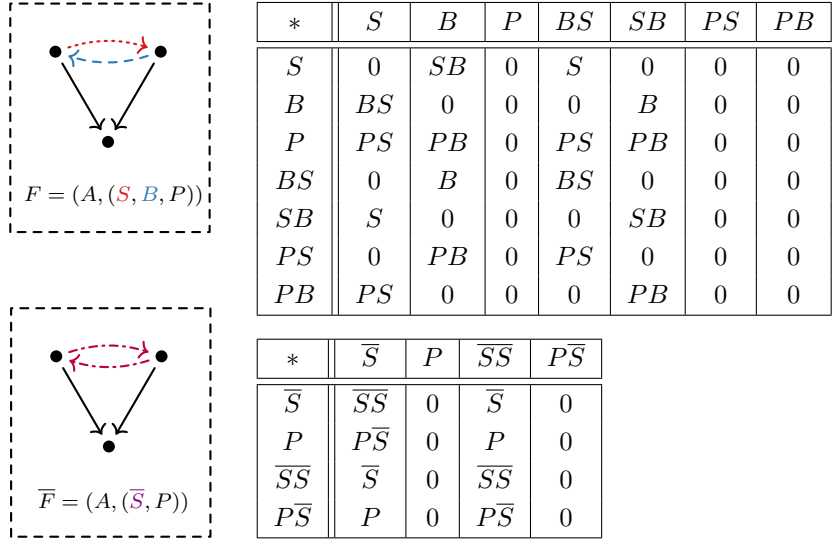
\begin{figure}
    \centering
    \begin{minipage}{0.3\textwidth}
        \centering
        \begin{tikzpicture}[line cap=round, line join=round, thick, yscale=0.8] 
        \tikzset{ dot/.style={circle, fill, inner sep=1.6pt}, } 
        \def\bend{22}
        \def\gap{2.5pt} 
    
        \node[dot] (a) at (-2.1, 2.2) {}; 
        \node[dot] (b) at (-0.7, 2.2) {}; 
        \node[dot] (d) at (-1.4, 0.7) {}; 

        \node[anchor=east] (g) at ( 0, -.2) {\footnotesize\(F = (A, ({\color{red} S}, {\color{blue}{B}}, {\color{black} P}))\)};
    
        \draw[->, shorten <=\gap, shorten >=\gap, black] (a) -- (d); 
        \draw[->, shorten <=\gap, shorten >=\gap, black] (b) -- (d); 
        
        \draw[->, shorten <=\gap, shorten >=\gap, red, dotted] (a) to[bend left=\bend] (b);

        \draw[->, shorten <=\gap, shorten >=\gap, blue, dashed] (b) to[bend left=\bend] (a);
        
        \draw[dashed] (-2.7, -.8) rectangle (-.05, 3); 
        \end{tikzpicture}

        \vspace{9.5mm}

        \begin{tikzpicture}[line cap=round, line join=round, thick, yscale=0.8] 
        \tikzset{ dot/.style={circle, fill, inner sep=1.6pt}, } 
        \def\bend{22}
        \def\gap{2.5pt} 
    
        \node[dot] (a) at (-2.1, 2.2) {}; 
        \node[dot] (b) at (-0.7, 2.2) {}; 
        \node[dot] (d) at (-1.4, 0.7) {}; 
        
        \node[anchor=east] (g) at ( -.2, -.2) {\footnotesize\(\overline{F} = (A, ({\color{violet} \overline{S}}, {\color{black} {P}}))\)};
    
        \draw[->, shorten <=\gap, shorten >=\gap, black] (a) -- (d); 
        \draw[->, shorten <=\gap, shorten >=\gap, black] (b) -- (d); 
        
        \draw[->, shorten <=\gap, shorten >=\gap, purple, dashdotted] (a) to[bend left=\bend] (b); 
        \draw[->, shorten <=\gap, shorten >=\gap, purple, dashdotted] (b) to[bend left=\bend] (a);
        
        \draw[dashed] (-2.7, -.8) rectangle (-.05, 3); 
        \end{tikzpicture}
    \end{minipage}
    \hfill 
    \begin{minipage}{0.69\textwidth}
        \renewcommand{\arraystretch}{1.2}
        \begin{tabular}{| c || c | c | c | c | c | c | c |}
        \hline
        \(\ast\) & \(S\)& \(B\) & \(P\) & \(BS\)& \(SB\) & \(PS\)& \(PB\) \\ \hline \hline
        \(S\)& 0 & \(SB\) & 0 & \(S\)& 0 & 0 & 0 \\
        \(B\) & \(BS\)& 0 & 0 & 0 & \(B\) & 0 & 0 \\
        \(P\) & \(PS\)& \(PB\) & 0 & \(PS\)& \(PB\) & 0 & 0 \\
        \(BS\)& 0 & \(B\) & 0 & \(BS\)& 0 & 0 & 0 \\
        \(SB\) & \(S\)& 0 & 0 & 0  & \(SB\) & 0 & 0 \\
        \(PS\)& 0 & \(PB\) & 0 & \(PS\)& 0 & 0 & 0 \\
        \(PB\) & \(PS\)& 0 & 0 & 0 & \(PB\) & 0 & 0 \\
        \hline
        \end{tabular}

        \vspace{2.7mm}
        
        \begin{tabular}{| c || c | c | c | c |}
        \hline
        \(\ast\) & \(\overline{S}\) & \({P}\) & \(\overline{S}\overline{S}\) & \({P}\overline{S}\) \\ \hline \hline
        \(\overline{S}\) & \(\overline{S}\overline{S}\) & 0 & \(\overline{S}\) & 0 \\
        \({P}\) & \({P}\overline{S}\) & 0 & \({P}\) & 0 \\
        \(\overline{S}\overline{S}\) & \(\overline{S}\) & 0 & \(\overline{S}\overline{S}\) & 0 \\
        \({P}\overline{S}\) & \({P}\) & 0 & \({P}\overline{S}\) & 0
        \\
        \hline
        \end{tabular}
    \end{minipage}

    \caption{A role reduction in which the `sister' and `brother' relations \(S\) and \(B\) are identified within the `sibling' relation \(\overline{S}\). See \Cref{eg:parent-sibling}.}
    \label{fig:parent-sibling}
\end{figure}

\begin{remark}
\label{R:non-associative}
    All compositions of relations that we consider in this paper are associative: they satisfy \((R_3 \ast R_2) \ast R_1 = R_3 \ast (R_2 \ast R_1)\) for all \(R_1,R_2,R_3\). However, there are contexts where the order in which relations are composed may alter the outcome---for instance, when the interpretation of a given kinship relation is gender-dependent. This is the case for the different roles that cousins play in some cultures, as discussed by Lévi-Strauss \cite{levi1971elementary}, who distinguishes between cousins that are considered siblings (the children of the mother's sisters and father's brothers) and cousins who may be spouses (the children of the mother's brothers and father's sisters). We leave the study of such types of composition for future work; we discuss related open questions in Section \ref{sec:future_work}. 
\end{remark}


\subsection{Compatibility of positional and role analysis}
\label{subsec:graph-functoriality}

Positional analysis and role analysis emerged over the same period in the 1970s and early 1980s, and naturally researchers were interested in performing both types of analysis on the same sets of data. Many early papers are concerned with the question of when the two forms of simplification---blockmodels and role reductions---are `compatible' in the sense that one determines the other; see, for example, \cite{Sailer1978} and the related discussion in \cite[pp.353-354]{Wasserman_Faust_1994}.

One version of this question can be formulated as follows. Given a \(k\)-relational graph \(G = (A, (R_i)_{i=1}^k)\), any equivalence relation \(E\) on \(A\) determines a simultaneous blockmodel of \(R_1,\ldots, R_k\): a map of \(k\)-relational graphs
\[G \twoheadrightarrow \block{G}{E} = (\block{A}{E}, (\block{R_i}{E})_{i=1}^k),\]
where, for each \(i\), the graph \((\block{A}{E}, \block{R_i}{E})\) is the blockmodel of \((A,R_i)\) by \(E\). The multirelational graphs \(G\) and \(\block{G}{E}\) each possess a semigroup of roles; the question is whether there exists a corresponding role reduction
\[\Role(G) \twoheadrightarrow \Role(\block{G}{E}).\]
Certainly, some simultaneous blockmodels do induce role reductions.

\begin{example}\label{eg:generations}
    Consider the 2-relational graph \(\overline{F}\) in \Cref{fig:parent-sibling} and its blockmodel \(\block{\overline{F}}{E}\) in \Cref{fig:generations}; the positions here are the `generations' of this small family. The quotient \(\overline{F} \twoheadrightarrow \block{\overline{F}}{E}\) induces a role reduction \(\Role(\overline{F}) \twoheadrightarrow \Role(\block{\overline{F}}{E})\), determined by \(\overline{S} \mapsto \block{\overline{S}}{E}\) and \(P \mapsto \block{P}{E}\); it identifies the role of `parent' with that of `sibling's parent'. 
\end{example}

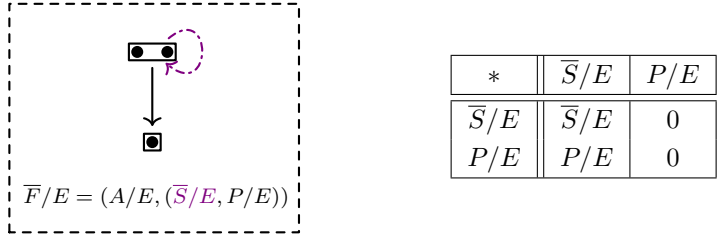
\begin{figure}
    \centering
    \begin{minipage}{0.45\textwidth}
        \centering
        \begin{tikzpicture}[line cap=round, line join=round, thick, yscale=0.8] 
        \tikzset{ dot/.style={circle, fill, inner sep=1.6pt}, } 
        \def\bend{22}
        \def\gap{2.5pt} 
    
        \node[dot] (a) at (-1.6, 2.2) {}; 
        \node[dot] (b) at (-1.2, 2.2) {}; 
        \node[dot] (d) at (-1.4, 0.7) {}; 
        \node[rectangle,draw,fit=(a) (b),inner sep=.5] (ab) {};
        \node[rectangle,draw,fit=(d),inner sep=.5] (dd) {};
        
        \node[anchor=east] (g) at ( .55, -.2) {\footnotesize\(\block{\overline{F}}{E} = (\block{A}{E}, ({\color{violet} \block{\overline{S}}{E}}, {\color{black} \block{P}{E}}))\)};
    
        \draw[->, shorten <=\gap, shorten >=\gap, black] (ab) -- (dd); 
        \path[->] (ab) edge [shorten <=\gap, shorten >=\gap, out=45, in=315, looseness=4, loop, violet, dashdotted] (ab);
        
        \draw[dashed] (-3.3, -.8) rectangle (.55, 3); 
        \end{tikzpicture}
    \end{minipage}
    \hfill 
    \begin{minipage}{0.45\textwidth}
        \renewcommand{\arraystretch}{1.2}
        
        \begin{tabular}{| c || c | c |}
        \hline
        \(\ast\) & \(\block{\overline{S}}{E}\) & \(\block{P}{E}\) \\ \hline \hline
        \(\block{\overline{S}}{E}\) & \(\block{\overline{S}}{E}\) & 0 \\
        \(\block{P}{E}\) & \(\block{P}{E}\) & 0 \\
        \hline
        \end{tabular}
    \end{minipage}

    \caption{A blockmodel of \(\overline{F}\) (\Cref{fig:parent-sibling}) into `generations', inducing a role reduction that identifies `parent' with `sibling's parent'. See \Cref{eg:generations}.}
    \label{fig:generations}
\end{figure}

Yet, this is not always the case, as the next example shows.

\begin{example}\label{eg:not-all-blockmodels}
    Consider the graph \(G = (A,R)\) in which \(A = \{a,a',b,b'\}\) and \(R = \{(b,a), (a',b')\}\). Let \(E \subseteq A \times A\) be the equivalence relation generated by \(\{(a,a'), (b,b')\}\). The blockmodel \(\block{G}{E}\) has two vertices, \(X = \{a,a'\}\) and \(Y = \{b,b'\}\), and its edge set is \(\block{R}{E} = \{(X,Y),(Y,X)\}\). 
    
    Since \(R \ast R = 0\), the semigroup \(\Role(G)\) has two elements, \(R\) and \(0\). Since \((\block{R}{E}) \ast (\block{R}{E}) = \delta_{\block{A}{E}}\), the semigroup \(\Role(\block{G}{E})\) also has two elements, \(\block{R}{E}\) and \(\delta_{\block{A}{E}} =: \delta\). In this case, there is no role reduction from \(\Role(G)\) to \(\Role(\block{G}{E})\). For, suppose a semigroup homomorphism \(\phi \colon \Role(G) \to \Role(\block{G}{E})\) were to satisfy \(\phi(R) = \block{R}{E}\); then it would satisfy
    \[\phi(0) = \phi(R \ast R) = \phi(R) \ast \phi(R) = (\block{R}{E}) \ast (\block{R}{E}) = \delta.\]
    But in that case we would have
    \[\phi(E \ast 0) = \phi(0) = \delta \neq (\block{R}{E}) \ast \delta = \phi(R) \ast \phi(0),\]
    contradicting that \(\phi\) is a homomorphism. Meanwhile, if \(\phi(R) = \delta\) then
    \[\phi(0) = \phi(R \ast R) = \phi(R) \ast \phi(R) = \delta \ast \delta = \delta\]
    so this semigroup homomorphism \(\phi\) is not surjective.
\end{example}

The question, then, is what conditions one might place on \(E\) to ensure that \(G \twoheadrightarrow \block{G}{E}\) does induce a role reduction. Notice that the equivalence relation \(E\) in \Cref{eg:generations} is outward-regular, while that in \Cref{eg:not-all-blockmodels} is neither outward- nor inward-regular. As it turns out, either outward- or inward-regularity with respect to each of the graphs \((A,R_i)\), is sufficient to ensure that the simultaneous blockmodel of \(G = (A, (R_i)_{i=1}^k)\) by \(E\) induces a role reduction.

This was understood within social network analysis by the early 1980s; for instance, it is stated without proof by Pattison in \cite[pp.91--92]{Pattison1982} as well as in \cite[p.181]{Pattison1993}, where it is attributed to  J.~Q.~Johnson. (Pattison calls the conditions of outward- and inward-regularity the `out-degree condition' and the `in-degree condition'.) To lay groundwork for the rest of the paper, we will give a formal statement of this folklore result. In due course we will see that \Cref{prop:induced-map-graph-case} can be derived from our \Cref{thm:functoriality-coalgebras}.

\begin{proposition}
\label{prop:induced-map-graph-case}
    Let \(G = (A, (R_i)_{i=1}^k)\) be a multirelational graph and let \(E \subseteq A \times A\) be an outward-regular equivalence on \((A,R_i)\) for all \(i \in \{1,\ldots,k\}\). Then there is a role reduction
    \[\phi \colon \Role(G) \twoheadrightarrow \Role\left(\block{G}{E}\right)\]
    defined by \(\phi(R) = \block{R}{E}\).
    The same holds if \(E\) is an inward-regular equivalence on \((A,R_i)\) for all \(i \in \{1,\ldots,k\}\).
\end{proposition}

The proposition motivates us to give a name to the map that takes the quotient of a multirelational graph by an outward-regular equivalence. (From here on, for brevity, we will make various definitions and statements in terms of outward-regularity only; each of these could be made in terms of inward-regularity instead.) First it will be helpful to characterize these maps concretely, for which we use the following definition.

\begin{definition}
\label{def:reflect-outgoing-edges}
    A map \(f \colon (A,(R_i)_{i=1}^k) \to (A',(R_i')_{i=1}^k)\) in \(k\Graph\) will be said to \define{reflect outgoing edges} if, for every vertex $a \in A$ and any edge $(f(a), b) \in R_i'$, there exists $a'\in A$ such that $(a, a') \in R_i$ and $f(a')=b$.
\end{definition}

For conciseness, we omit the proof of the next lemma. In due course, it can be recovered as a special case of \Cref{lem:coalgebra-PR-concretely}.

\begin{lemma}
\label{lem:graph-PR-concretely}
    For a map \(f \colon (A, (R_i)_{i=1}^k) \to (A', (R_i')_{i=1}^k)\) in \(k\Graph\), the following are equivalent.
    \begin{itemize}
        \item The equivalence relation \(E = \{(a,a') \mid f(a) = f(a')\}\) is outward-regular on \((A,R_i)\) for each \(i \in \{1,\ldots, k\}\), and \((A',R_i') = (\block{A}{E}, \block{R_i}{E})\); the function $f$ sends each vertex \(a \in A\) to its equivalence class in \(\block{A}{E}\).
        \item The function $f \colon A \to A'$ is surjective, and \(f\) reflects outgoing edges.
    \end{itemize}
\end{lemma}

On the strength of the lemma, we make the following definition.

\begin{definition}
\label{def:positional-reduction-graph}
    A \define{positional reduction} of a \(k\)-relational graph \(G\) is a map \(G \twoheadrightarrow H\) in \(k\Graph\) which is surjective on vertices and reflects outgoing edges.
\end{definition}

In practice, positional analysis typically involves a sequence of positional reductions, grouping actors together in stages into a nested family of blockmodels \cite[\S9.5.1]{Wasserman_Faust_1994}. One would like to know that the induced role reductions behave consistently. Otter and Porter \cite{OtterPorter2020} frame this consistency in terms of functoriality, and that is the approach we will adopt throughout this paper. 

To state the functoriality theorem for graphs, let us first observe that the composite of two positional reductions is again a positional reduction, so there is a category \(k\Graph_{\mathsf{PR}}\) of \(k\)-relational graphs and positional reductions. We denote by \(\mathbf{SGrp}_\mathsf{Q}\) the category of semigroups and their quotient maps, which are precisely surjective semigroup homomorphisms.

\begin{theorem}
\label{thm:functoriality-graphs}
    The assignment of the semigroup of roles extends to a functor
    \[\Role\colon k\Graph_{\mathsf{PR}} \to \mathbf{SGrp}_\mathsf{Q}.\]
\end{theorem}

\begin{proof}
    \Cref{prop:induced-map-graph-case} says that every map \(q \colon G \twoheadrightarrow H\) in \(k\Graph_{\mathsf{PR}}\) induces a map \(\Role(q) \colon \Role(G) \twoheadrightarrow \Role(H)\) in \(\mathbf{SGrp}_\mathsf{Q}\). Explicitly, writing \(V(G)\) for the vertex set of \(G\), the map \(\Role(q)\) takes each graph in \(\Role(G)\) to its blockmodel with vertex set 
    \[V(H) = \block{V(G)}{\{(a,a') \mid q(a) = q(a')\}}.\]
    Thus, \(\Role(\mathrm{Id}_G) = \mathrm{Id}_{\Role(G)}\) and, given \(q \colon G \twoheadrightarrow H\) and \(p \colon H \twoheadrightarrow K\), the homomorphisms \(\Role(p \circ q)\) and \(\Role(p) \circ \Role(q)\) both take each graph in \(\Role(G)\) to its blockmodel with vertex set \(V(K)\). That is, \(\Role(p \circ q) = \Role(p) \circ \Role(q)\).
\end{proof}


\section{Social systems as coalgebras}\label{sec:ss_coalg}

The techniques of positional and role analysis described in \Cref{sec:positional_role} are well-established for social systems modelled by graphs; Wasserman and Faust \cite[Part IV]{Wasserman_Faust_1994} give a comprehensive account of how they can be applied in practice. However, there are many types of social system for which graphs do not provide a sufficiently rich model---see, for example, \cite{BONACICH2004189}. Our goal is  to extend positional and role analysis to systems that involve not just pairwise relations among actors, but higher-order relations too.  Our work inscribes itself in a line of current efforts devoted to generalising methods developed for graphs to higher-order models \cite{BGHS, TBBE21}.

The key to our approach is to interpret a social system not as something static but as a form of \emph{transition system}, and represent it by a category-theoretic gadget known as a \emph{coalgebra}. The purpose of this section is to explain what that means. First, in \Cref{subsec:higher-order-relations}, we will describe how higher-order relations can be represented by a hypergraph. (For a much more comprehensive overview, see \cite{TBBE21}.) Then, in \Cref{subsec:coalgebras} we introduce coalgebras and explain how they can they can be used to encode both graphs and hypergraphs.

Throughout this section and the rest of the paper we will make repeated reference to the covariant powerset functor, which we now introduce formally.

\begin{definition}\label{def:powerset}
    For each set \(X\), let \(\power(X)\) denote the \define{powerset} of \(X\); that is,
    \[\power(X) = \{U \mid U \subseteq X\}.\]
    The \define{(covariant) powerset functor} is the functor \(\power \colon \Set \to \Set\) taking each set to its powerset and each function \(f \colon X \to Y\) to the \define{direct image} function
    \begin{align*}
        \power(f) \colon \power(X) &\to \power(Y) \\
        U &\mapsto f(U) = \{f(u) \mid u \in U\}.
    \end{align*}
    We will refer to the functor \(\power\power \colon \Set \to \Set\) as the \define{double powerset functor}.
\end{definition}

\subsection{Social systems as hypergraphs}
\label{subsec:higher-order-relations}

\begin{figure}
    \centering
    \begin{tikzpicture}[line cap=round, line join=round, thick, yscale=0.8] 
    \tikzset{ dot/.style={circle, fill, inner sep=1.6pt}, 
    over/.style={preaction={draw=white, line width=\pgflinewidth+1.2pt}},}
    \def\bend{16}
    \def\gap{2.5pt} 

    \node[] (a) at (-1.5, 1.7) {a}; 
    \node[] (b) at (-0.2, 1.7) {b}; 
    \node[] (c) at (-2.7, 0.7) {c}; 
    \node[] (d) at (-1.5, -0.3) {d}; 
    \node[] (e) at (-0.2, -0.3) {e}; 

    \node[anchor=east] (g) at ( -2.5, -0.9) {\footnotesize $(1)$}; 

    \draw[shorten <=\gap, shorten >=\gap, black] (a) -- (b); 
    \draw[shorten <=\gap, shorten >=\gap, black] (a) -- (c); 
    \draw[shorten <=\gap, shorten >=\gap, black] (c) -- (d); 
    \draw[shorten <=\gap, shorten >=\gap, black] (a) -- (d); 
    \draw[shorten <=\gap, shorten >=\gap, black] (b) -- (d); 
    \draw[shorten <=\gap, shorten >=\gap, black] (e) -- (b); 
    \draw[shorten <=\gap, shorten >=\gap, black] (d) -- (e); 
    
    \draw[dashed] (-3.2,-1.3) rectangle (0.4,2.7); \end{tikzpicture}
    \quad
    \begin{tikzpicture}[line cap=round, line join=round, thick, yscale=0.8,
    he/.style={draw, semithick},        
    ]
    \tikzset{ dot/.style={circle, fill, inner sep=1.6pt}, } 
    \def\bend{20}
    \def\gap{2.5pt} 

    \node[] (a) at (-1.5, 1.7) {a}; 
    \node[] (b) at (-0.2, 1.7) {b}; 
    \node[] (c) at (-2.7, 0.7) {c}; 
    \node[] (d) at (-1.5, -0.3) {d}; 
    \node[] (e) at (-0.2, -0.3) {e}; 

    \node[anchor=east] (g) at ( -2.5, -0.9) {\footnotesize $(2)$}; 

    \draw[he] \hedgeii{a}{b}{3mm};
    \draw[he] \hedgeiii{a}{d}{c}{3mm};
    \draw[he] \hedgeiii{b}{e}{d}{3mm};
    
    \draw[dashed] (-3.2,-1.3) rectangle (0.4,2.7); 
    \end{tikzpicture}
    \quad
    \begin{tikzpicture}[line cap=round, line join=round, thick, yscale=0.8] 
    \tikzset{ dot/.style={circle, fill, inner sep=1.6pt}, 
    over/.style={preaction={draw=white, line width=\pgflinewidth+1.2pt}},} 
    \def\bend{20}
    \def\gap{.5pt} 

    \node[] (a) at (-1.5, 1.7) {a}; 
    \node[] (b) at (-0.2, 1.7) {b}; 
    \node[] (c) at (-2.7, 0.7) {c}; 
    \node[] (d) at (-1.5, -0.3) {d}; 
    \node[] (e) at (-0.2, -0.3) {e}; 

    \node[anchor=east] (g) at ( -2.5, -0.9) {\footnotesize $(3)$}; 

    \coordinate (ab) at (-0.8, 1.7); 
    \coordinate (abd) at (-2.1, 0.7); 
    \coordinate (bde) at (-0.45, 0.7);

    \draw[->, shorten <=\gap] (a) -- (ab); 
    \draw[->, shorten <=\gap] (c) -- (abd); 
    \draw[->, shorten <=\gap] (b) -- (bde);

    \draw[shorten >=\gap] (ab) -- (b); 
    \draw[shorten >=\gap] (abd) to[bend right=\bend] (a); 
    \draw[shorten >=\gap] (abd) to[bend left=\bend] (d); 
    \draw[shorten >=\gap] (bde) to[bend left=\bend] (d); 
    \draw[shorten >=\gap] (bde) to[bend right=\bend] (e); 
    
    \draw[dashed] (-3.2,-1.3) rectangle (0.4,2.7); 
    \end{tikzpicture}
    
    \caption{Three models of the same co-authorship system. See \Cref{subsec:higher-order-relations}.}
    \label{fig:coauthorship}
\end{figure}
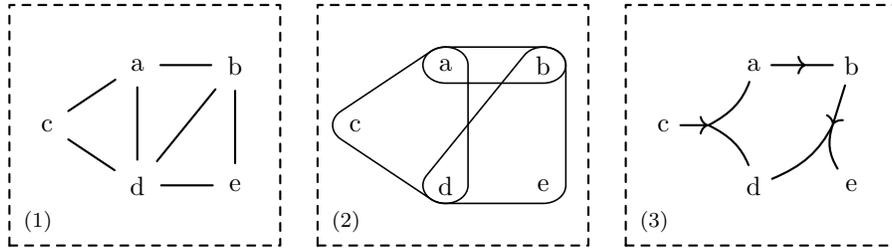

As an example of a social system involving higher-order relations, consider a \emph{co-authorship system}. A basic co-authorship system comprises a set \(A\) of authors and a `bibliography' \(B\): a set of papers written by the authors in \(A\). A good model of the system should faithfully represent the co-authorship relations among the authors. Co-authorship systems have often been thought of as networks and modelled by undirected graphs (e.g.~\cite{Newman2004})---however, such a model is typically {not} a faithful representation.

To be concrete, suppose we are given a set of five authors, \(A = \{a,b,c,d,e\}\), and a set of three papers, \(B = \{P_1,P_2,P_3\}\). Paper \(P_1\) has authors \(a,c,d\), paper \(P_2\) has authors \(a, b\), and \(P_3\) has authors \(b,c,e\). Conventionally, the co-authorship system \((A,B)\) might be modelled by the undirected graph in \Cref{fig:coauthorship}(1), which has vertex set \(A\) and an edge between two vertices if there exists a paper in \(B\) on which they are both authors. This model captures the pairwise co-authorship relations but not the {higher-order} ones: given a triple of vertices linked pairwise by edges, there is no way to tell whether or not there is a paper in \(B\) which carries all three as authors. To capture these higher-order relations one should use, instead of a graph, a \emph{hypergraph}.

\begin{definition}\label{def:undirected-hypergraph}
    A \define{hypergraph} is a pair of sets \((A, H)\) in which \(H \subseteq \power(A)\). The elements of \(A\) are called \define{vertices} and elements of \(H\) are called \define{hyperedges}. A \define{map of hypergraphs} \((A,H) \to (A',H')\) is a function \(f \colon A \to A'\) such that, for every \(U \in H\), the direct image \(f(U)\) belongs to \(H'\). We denote the category of hypergraphs and these maps by \(\HGraph\).
\end{definition}

We depict hypergraphs as in \Cref{fig:coauthorship}(2): the five vertices are the authors in \(A\) and each of the three papers in \(B\) is represented by a single hyperedge containing its authors. This hypergraph is a faithful representation of the pair of sets \((A,B)\). In some contexts, though, it might still not be a sufficiently rich model of the co-authorship system. Suppose the authors in \(A\) happen to belong to a discipline in which authorship {order} is important; in that case one might like the model to record, for each paper, who is first author. For this, one should use a \emph{directed hypergraph}.

\begin{definition}\label{def:directed-hypergraph}
    A \define{directed hypergraph} is a pair of sets \(\mathcal{H} = (A,H)\) in which \(H \subseteq \power(A) \times \power(A)\). Again, the elements of \(H\) will be called \define{hyperedges}. In a hyperedge \((U,V)\), the set \(U\) is called the \define{source} and \(V\) is called the \define{target}.
\end{definition}

We could model our co-authorship system by a directed hypergraph whose hyperedges are pairs of the form \((\{\text{first author of \(P_i\)}\}, \{\text{other authors of \(P_i\)}\})\) for \(i=1,2,3\). This directed hypergraph has the special property that the source of every hyperedge is a singleton; such structures are known as \emph{forward hypergraphs} or \emph{F-hypergraphs}.

\begin{definition}\label{def:f-hypergraph}
    An \define{F-hypergraph} is a directed hypergraph in which the source of every hyperedge is a singleton set. In an F-hypergraph, we will denote a hyperedge \((\{a\},U)\) by \((a,U)\). A \define{map of F-hypergraphs} \((A,H) \to (A',H')\) is a function \(f \colon A \to A'\) such that for every \((a,V) \in H\) we have \((f(a),f(V)) \in H'\). We denote the category of these by \(\FHGraph\).
\end{definition}

We will depict F-hypergraphs as in \Cref{fig:coauthorship}(3), with arrows pointing outward from the source of each hyperedge. This F-hypergraph records that the first author of \(P_1\) is \(c\), the first author of \(P_2\) is \(a\), and the first author of \(P_3\) is \(b\). 

F-hypergraphs can be regarded as a strict generalization of hypergraphs. Indeed, \(\HGraph\) embeds as a a full subcategory of \(\FHGraph\).

\begin{proposition}\label{prop:hypergraphs-F-hypergraphs}
    There is a full and faithful functor
    \(\HGraph \hookrightarrow \FHGraph\).
\end{proposition}

\begin{proof}
    Every hypergraph \(\mathcal{H} = (A,H)\) gives rise to an F-hypergraph \(F(\mathcal{H})\) with vertex set \(A\) and hyperedge set \(\{ (a, U) \mid \{a\} \sqcup U \in H\}\). It is easily seen that a function \(f \colon A \to A'\) is a map of hypergraphs \(\mathcal{H} \to \mathcal{H}'\) if and only if it is a map of F-hypergraphs \(F(\mathcal{H}) \to F(\mathcal{H}')\); thus, \(\mathcal{F}\) defines a full and faithful functor.
\end{proof}

Note that not every F-hypergraph is in the image of this subcategory inclusion. For instance, the F-hypergraph in \Cref{fig:coauthorship}(3) cannot be \(F(\mathcal{H})\) for any hypergraph \(\mathcal{H}\), since, for example, \((a,\{b\})\) is a hyperedge but \((b,\{a\})\) is not. This means that F-hypergraphs are strictly more general than hypergraphs and can describe a richer variety of higher-order relations. 

The goal of this paper is to define positional and role analysis for these higher-order relations. We will make our definitions for F-hypergraphs, and by \Cref{prop:hypergraphs-F-hypergraphs} they can be specialized to hypergraphs. Our model of `multirelational higher-order network', generalizing a multirelational graph, will therefore be a collection of several F-hypergraph structures on the same set of vertices.

\begin{definition}
    A \define{\(k\)-relational F-hypergraph} is a pair \(\mathcal{H} = (A, (H_i)_{i=1}^k)\) of a set \(A\) and a tuple \((H_1,\ldots,H_k)\) such that each \((A,H_i)\) is an F-hypergraph. A \define{map of \(k\)-relational F-hypergraphs}
    \[(A, (H_i)_{i=1}^k) \to (A', (H_i')_{i=1}^k)\]
    is a function \(A \to A'\) which is a map of directed hypergraphs \((A,H_i) \to (A',H_i')\) for every \(i\). We denote the category of these by \(k\FHGraph\).
\end{definition}

\subsection{Graphs and hypergraphs as coalgebras}
\label{subsec:coalgebras}

Coalgebras emerged after the 1970s in a category-theoretic approach to systems theory whose goal was to develop a unified framework for sequential machines and control systems. As a comprehensive introduction to the field, we recommend Rutten \cite{R20}.

The canonical example of a coalgebra is a \emph{non-deterministic transition system}. Such a system consists of a set \(X\) of \emph{states} and a function assigning to each \(x \in X\) a set of \emph{possible next states} \(\alpha(x) \subseteq X\). That is, it is a function 
\[\alpha \colon X \to \power(X).\]
A non-deterministic transition system can be represented graphically by drawing a  directed edge from each state to each of its possible next states. For instance, if $X=\{x_0,x_1,x_2\}$ and \(\alpha\) is given by $\alpha(x_0)=\{x_1,x_2\}$, $\alpha(x_1)=\{x_0,x_1\}$ and $\alpha(x_2)=\{x_1\}$, this produces the graph
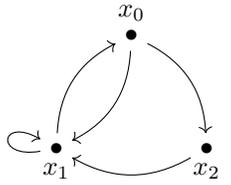
\[
\begin{tikzpicture}
\node (x0) at (0,0.5) {$\bullet$};
\node (x1) at (-1,-1) {$\bullet$};
\node (x2) at (1,-1) {$\bullet$};
\node at (-1,-1.3) {$x_1$};
\node at (1,-1.3) {$x_2$};
\node at (0,0.8) {$x_0$};
\path[->] 
(x0) edge [bend left] (x1)
(x0) edge [bend left] (x2)
(x1) edge [bend left] (x0)
(x2) edge [bend left] (x1);
\draw[->] (x1) to [out=190,in=150,looseness=10] (x1);
\end{tikzpicture}.
\]
Conversely, every graph \(G = (A,R)\) gives rise to a non-deterministic transition system \(A \to \power(A)\) which assigns to each vertex its `out-neighbourhood'---the set of all vertices to which it admits an edge.

Coalgebras generalize non-deterministic transition systems by allowing the set \(\power(A)\) to be replaced by \(T(A)\) for any functor \(T \colon \Set \to \Set\).

\begin{definition}\label{def:T-coalgebra}
    Given a functor \(T \colon \Set \to \Set\), a \define{\(T\)-coalgebra} is a pair \((A, \alpha)\) of a set \(A\) and a function \(\alpha \colon A \to T(A)\). A \define{homomorphism of \(T\)-coalgebras} $(A, \alpha) \to (A', \alpha')$ is a function $f \colon A \to A'$ making this diagram commute:
    \[
    \begin{tikzcd}
    A \arrow[d, swap, "\alpha"] \arrow[r, "f"] & A' \arrow[d, "\alpha'"]\\ 
    T(A) \arrow[r, swap, "T(f)"] & T(A').
    \end{tikzcd}
    \]
    We denote by $\Coalg(T)$ the category of \(T\)-coalgebras and their homomorphisms, in which identity maps and composition are inherited from $\Set$.
\end{definition}

With suitable choices of $T$, coalgebras recover many different kinds of `systems', including automata and infinite streams \cite{R20} and probabilistic transition systems \cite{BartelsSokolovaVink2004}. One can also talk about coalgebras for a functor \(T \colon \mathbf{C} \to \mathbf{C}\) where \(\mathbf{C}\) is some category other than $\Set$; in this way, one can obtain coalgebraic treatments of certain dynamical systems \cite{R20} and hybrid systems combining discrete updates with continuous time evolution \cite{Jacobs2000Hybrid}. In this paper, however, the most important examples besides \(\power\)-coalgebras will be coalgebras for the \emph{double} covariant powerset functor, \(\power\power\).

\begin{example}[\(\power\power\)-coalgebras are F-hypergraphs]
\label{eg:hypergraph-coalgebras}
    Just as every graph can be specified by its out-neighbourhood function, every F-hypergraph \(\mathcal{H} = (A,H)\) can be specified by its neighbourhood function
    \[
    \begin{aligned}
    \mathcal{N}_\mathcal{H} \colon A &\to \power\power(A) \\
    v &\mapsto \{U \mid (v, U) \in H\} .
    \end{aligned}
    \]
    Conversely, every coalgebra for the double powerset functor, \(\alpha \colon A \to \power\power(A)\), can be specified by the F-hypergraph \(\mathcal{H}_\alpha = (A, H_\alpha)\) with hyperedge set 
    \[H_\alpha = \{(v,U) \mid U \in \alpha(v)\}.\]
    Thus, for each set \(A\) there is a bijection between F-hypergraph structures on \(A\) and \(\power\power\)-coalgebra structures on \(A\).
\end{example}

In Sections \ref{sec:positional_bisim} and \ref{sec:role_kleisli} we will extend positional and role analysis from graphs to \(T\)-coalgebras for an arbitrary functor \(T \colon \Set \to \Set\). In particular, by \Cref{eg:hypergraph-coalgebras}, this will define positional and role analysis for F-hypergraphs. Our main object of interest will be a \emph{multirelational \(T\)-coalgebra}: a collection of several coalgebra structures supported on the same set.

\begin{definition}\label{def:multirelational-coalgebras}
    Given a functor \(T \colon \Set \to \Set\), a \define{\(k\)-relational \(T\)-coalgebra} is a pair \((A, (\alpha_i)_{i=1}^k)\) of a set \(A\) and a tuple \((\alpha_1, \ldots, \alpha_k)\) of \(T\)-coalgebra structures on \(A\). A \define{map of \(k\)-relational \(T\)-coalgebras}
    \[(A, (\alpha_i)_{i=1}^k) \to (A', (\alpha_i')_{i=1}^k)\]
    is a function \(A \to A'\) which is a \(T\)-coalgebra homomorphism \((A,\alpha_i) \to (A',\alpha_i')\) for each \(i \in \{1,\ldots,k\}\). We denote the category of these by \(k\Coalg(T)\).
\end{definition}

We have seen that every \(\power\)-coalgebra \(\alpha \colon A \to \power(A)\) determines a graph \(G_\alpha^{\mathrm{out}} = (A,R)\) with
\(E = \{(a,b) \mid b \in \alpha(a)\}\), and every graph \(G = (A,R)\) defines a \(\power\)-coalgebra \(N_G^{\mathrm{out}}\colon A \to \power(A)\) with \(N_G^{\mathrm{out}}(a) = \{b \in A \mid (a, b) \in E\}\).
This gives, for each set \(A\), a bijection between \(\power\)-coalgebra structures and graph structures on \(A\). It follows that there is a bijection between the objects of \(k\Coalg(\power)\) and those of \(k\Graph\). However, this does not extend to an equivalence of categories: not every map of graphs is a \(\power\)-coalgebra homomorphism.

\begin{example}[Homomorphisms of \(\power\)-coalgebras]
\label{eg:P-coalgebra-homs}
    Given graphs $G=(A,R)$ and $H=(A',R')$, a homomorphism of \(\power\)-coalgebras $(A,N_G^\mathsf{out}) \to (A',N_{H}^\mathsf{out})$ is a function $f \colon A \to A'$ satisfying
    \(
    \power(f)\circ N_G^\mathsf{out} = N_{H}^\mathsf{out} \circ f.
    \)
    That is, \(f \colon A \to A'\) is a homomorphism if and only if it satisfies
    \[\{f(a') \mid (a,a') \in R\} = \{b \mid (f(a), b) \in R'\}\]
    for every \(a \in A\). The inclusion of sets from left to right here says that \(f\) preserves edges---that is, \(f\) is a map of graphs. The inclusion of sets from right to left says that \(f\) reflects outgoing edges in the sense of \Cref{def:reflect-outgoing-edges}.
\end{example}

\Cref{eg:P-coalgebra-homs} tells us that every coalgebra homomorphism is a map of graphs, but not every map of graphs is a coalgebra homomorphism. In category-theoretic parlance, \(\power\)-coalgebras and their homomorphisms form a \emph{wide} subcategory of \(\Graph\), but not a \emph{full} subcategory. It follows that the same holds for \(k\Coalg(\power)\) and \(k\Graph\).

\begin{proposition}\label{prop:CoalgP-Graph-functor}
    The assignment
    \begin{align*}
        G_{(-)}^\mathsf{out} \colon k\Coalg(\power) &\to k\Graph \\
        (A, (\alpha_i)_{i=1}^k) & \mapsto (A, (G_{\alpha_i}^\mathsf{out})_{i=1}^k)
    \end{align*}
    determines a functor which is bijective on objects and faithful, but not full. \qed
\end{proposition}

\begin{remark}\label{rem:in-neighbourhood}
    There is a second functor exhibiting \(k\Coalg(\power)\) as a wide subcategory of \(k\Graph\). This one takes a \(\power\)-coalgebra \(\alpha\) to the graph \(G_\alpha^{\mathrm{in}} = (A,R)\) with \(R = \{(b,a) \mid b \in \alpha(a)\}\); in the other direction, there is a function on objects taking a graph \(G\) to its `in-neighbourhood' coalgebra 
    \[N_G^{\mathrm{in}} \colon a \mapsto \{b \in A \mid (b,a) \in R\}.\]
    Given graphs \(G = (A,R)\) and \(G' = (A',R')\), a coalgebra homomorphism $(A,N_G^\mathsf{in})\to (A',N_{G'}^\mathsf{in})$ is a function $f\colon A\to A'$ that preserves edges and \emph{reflects incoming edges}, meaning that for every \(a \in A\) and any edge \((b,f(a)) \in R'\), there exists \(a' \in A\) such that \((a',a) \in R\) and \(f(a') = b\).
\end{remark}

An analogous statement can be made for \(\power\power\)-coalgebras and F-hypergraphs: \(k\Coalg(\power\power)\) is a wide subcategory of \(k\FHGraph\) but not a full subcategory. To characterize those maps of F-hypergraphs that do arise from maps of \(\power\power\)-coalgebras, we will use the following definition.

\begin{definition}
\label{def:reflect-hyperedges}
    Let \(\mathcal{H} = (A,(H_i)_{i=1}^k)\) and \(\mathcal{K} = (A',(H_i')_{i=1}^k)\) be \(k\)-relational F-hypergraphs. A function \(f \colon A \to A'\) will be said to \define{reflect hyperedges} if for every vertex $a \in A$ and any hyperedge $(f(a), U) \in H_i'$, there exists $V \subseteq A$ such that $(a, V) \in R_i$ and $f(V)=U$.
\end{definition}

\begin{proposition}\label{prop:CoalgPP-HGraph-functor}
    The assignment
    \begin{align*}
        \mathcal{H}_{(-)} \colon k\Coalg(\power\power) &\to k\FHGraph \\
        (A, (\eta_i)_{i=1}^k) & \mapsto (A, (H_{\eta_i})_{i=1}^k)
    \end{align*}
    determines a functor which is bijective on objects and faithful, but not full.
\end{proposition}

\begin{proof}
    From \Cref{eg:hypergraph-coalgebras} we can see that \(\mathcal{H}_{(-)}\) is a bijection, with inverse function \(\mathcal{N}_{(-)}\). Given F-hypergraphs \(\mathcal{H} = (A,H)\) and \(\mathcal{H}' = (A',H')\), a coalgebra homomorphism \(\mathcal{N}_\mathcal{H} \to \mathcal{N}_{\mathcal{H}'}\) is a function \(f \colon A \to A'\) satisfying
    \(
    \power\power(f)\circ N_\mathcal{H} = N_{\mathcal{H}'} \circ f.
    \)
    That is, \(f \colon A \to A'\) is a homomorphism if and only if it satisfies
    \[\{f(V) \mid (a,V) \in H\} = \{U \mid (f(a), U) \in H'\}\]
    for every \(a \in A\). The inclusion of sets from left to right here says that \(f\) preserves hyperedges---hence, every coalgebra homomorphism is a map of F-hypergraphs, so \(\mathcal{H}_{(-)}\) is a functor (and clearly faithful). The inclusion from right to left says that \(f\) reflects hyperedges in the sense of \Cref{def:reflect-hyperedges}. Not every map of F-hypergraphs has this property, hence \(\mathcal{H}_{(-)}\) is not full.
\end{proof}

In the next section we will see that an outward-regular equivalence is an \emph{internal equivalence relation}, not only in \(\Graph\) but in \(\Coalg(\power)\). Such equivalence relations play a central role in the study of coalgebra, where they are known as \emph{bisimulation equivalences}. The importance of outward-regular equivalences in the story of positional and role analysis suggests that the coalgebraic setting is a natural one in which to interpret those methods.


\section{Positional analysis via bisimulation}
\label{sec:positional_bisim}

In \Cref{subsec:positional} we described \emph{positional analysis} for a social system modelled by a graph. Positional analysis concerns the construction of \emph{blockmodels}; these are simply quotients in the category of graphs. We saw that certain blockmodels are more sensible than others: taking the quotient by an \emph{outward-regular equivalence} produces a blockmodel that captures more accurately the concept of a `position' in a social network and is compatible with the process of \emph{role analysis} described in \Cref{subsec:role}. We therefore defined a \emph{positional reduction} of a multirelational graph \(G = (A, (R_i)_{i=1}^k)\) to be a map that takes the quotient of the graphs \((A,R_1), \ldots, (A,R_k)\) by an outward-regular equivalence.

Now, we are going to lift each of these notions from graphs to \(T\)-coalgebras for an arbitrary functor \(T \colon \Set \to \Set\). In \Cref{subsec:coalg-bisimulation} we will observe that an outward-regular equivalence on a graph is precisely a \emph{bisimulation equivalence} on the corresponding out-neighbourhood \(\power\)-coalgebra. On the strength of that observation, in \Cref{subsec:coalg-blockmodelling} we will define a \emph{positional reduction} of a multirelational \(T\)-coalgebra \(\alpha = (A, (\alpha_i)_{i=1}^k)\) to be a map that takes the quotient of \((A,\alpha_1), \ldots, (A,\alpha_k)\) by a bisimulation equivalence. 

Specializing these notions to \(\power\power\)-coalgebras will tell us what regular equivalences, blockmodels and positional reductions ought to be for higher-order social systems modelled by F-hypergraphs. Each of those definitions will be given explicitly in \Cref{sec:analysis-for-hypergraphs}.


\subsection{Bisimulation equivalences}
\label{subsec:coalg-bisimulation}

In the case of graphs, blockmodelling can be described in category-theoretic terms as follows. Given a graph \(G = (A,R)\), we can make any equivalence relation \(E \subseteq A \times A\) into a full subgraph of the product \(G \times G\) by declaring that \(((a,b),(a',b'))\) is an edge in \(E\) if and only if \((a,a') \in R\) or \((b,b') \in R\). This makes the inclusion \(\iota \colon E \hookrightarrow G \times G\) into a \emph{congruence} (or \emph{internal equivalence relation}) in the category \(\Graph\), and the two composites 
\[p_1, p_2 \colon E \overset{\iota}{\hookrightarrow} G \times G \overset{\pi_1}{\underset{\pi_2}{\rightrightarrows}} G\]
into maps of graphs. The blockmodel \(\block{G}{E}\) is obtained by taking the coequalizer of \(p_1\) and \(p_2\) in \(\Graph\), also known as the {quotient} of the graph \(G\) by the congruence \(R\). (For an introduction to congruences and their quotients, see Borceux \cite[\S 2.5]{Borceux-HCA-2}.)

The blockmodel construction applies to any equivalence relation on the set of vertices of a graph, without conditions. In \Cref{subsec:graph-functoriality}, however, we saw that in order for positional analysis to be compatible with role analysis, one should restrict attention to the class of equivalence relations we termed \emph{outward-regular}. Now we are in a position to interpret that condition. 

Inspecting the definitions, one sees that an equivalence relation \(E\) on \(A\) is outward-regular if and only if \(\iota \colon E \hookrightarrow G \times G\) reflects outgoing edges in the sense of \Cref{def:reflect-outgoing-edges}. According to \Cref{eg:P-coalgebra-homs}, this says that \(\iota\) is not just a map of graphs but of \(\power\)-coalgebras, and it follows that
\[\iota \colon N^\mathrm{out}_E \hookrightarrow N^\mathrm{out}_G \times N^\mathrm{out}_G\]
is actually a congruence in the category \(\Coalg(\power)\). Thus, if \(E\) is outward-regular, then \(N^\mathrm{out}_{\block{G}{E}}\) is the quotient of \(N^\mathrm{out}_G\) by \(N^\mathrm{out}_E\) in \(\Coalg(\power)\). 

This is the category-theoretic recipe we will follow to construct blockmodels of systems modelled by coalgebras for an arbitrary functor \(T\). A \emph{\(T\)-blockmodel} will be the quotient of a \(T\)-coalgebra by a congruence in the category \(\Coalg(T)\). Congruences in \(\Coalg(T)\) are at the heart of the study of coalgebra, where they are known by the name \emph{bisimulation equivalences}. We refer to \cite{Staton2011} for an overview of bisimulations in coalgebra.

\begin{definition}\label{def:T-bisimulation-span}
    Let \(T \colon \Set \to \Set\) be a functor and let \(\alpha \colon A \to T(A)\) and \(\beta \colon B \to T(B)\) be \(T\)-coalgebras. A relation \(E \subseteq A \times B\) is a \define{\(T\)-bisimulation} between \((A,\alpha)\) and \((B,\beta)\) if there exists a \(T\)-coalgebra structure \(\epsilon \colon E \to T(E)\) such that the projections
    \(A \xleftarrow{p_A} E \xrightarrow{p_B} B\) 
    are \(T\)-coalgebra homomorphisms, meaning that this diagram commutes:
    \[
    \begin{tikzcd}
        A \arrow[d, swap, "\alpha"] & \arrow[l, swap, "p_A"] E \arrow[d, "\epsilon"] \arrow[r, "p_B"] & B \arrow[d, "\beta"]\\
        T(A) & \arrow[l, "T(p_A)"] T(E) \arrow[r, swap, "T(p_B)"] & T(B).
    \end{tikzcd}
    \]
    A \define{\(T\)-bisimulation equivalence} is a \(T\)-bisimulation \((E,\epsilon)\) between \((A,\alpha)\) and \((A,\alpha)\) such that \(E \subseteq A \times A\) is an equivalence relation.
\end{definition}

In order to instantiate \Cref{def:T-bisimulation-span} for choices of functor other than \(\power\), it is useful to be able to say concretely when an equivalence relation on a set \(A\) is a bisimulation equivalence with respect to a given coalgebra structure. For this, one can employ the notion of \emph{relation lifting}.

\begin{definition}\label{def:T-lifting}
    Let \(T \colon \Set \to \Set\) be a functor. Given a relation of sets $E \subseteq A \times B$, let \(A \xleftarrow{\pi_A} A \times B \xrightarrow{\pi_B} B\) be the projections. The \define{\(T\)-relation lifting} of \(E\) is the relation $\lift{T}(E) \subseteq T(A) \times T(B)$ defined by 
    \begin{equation*}
        \lift{T}(E) = \left\{ (x, y) \in T(A) \times T(B) \;\middle|\; 
        \begin{aligned}
        &\text{there exists } u \in T(E) \text{ such that } \\ 
        &T(\pi_A)(u) = x \text{ and } T(\pi_B)(u) = y
        \end{aligned}
        \right\}.
    \end{equation*}
\end{definition}

If we think of the functor \(T\) as taking a set \(A\) of possible states to a set \(T(A)\) of possible transitions, then \(T\)-relation lifting takes a relation \(E\) among states to a relation \(\lift{T}(E)\) among transitions. The following theorem expresses that under a bisimulation equivalence, two states are equivalent only if the transitions they can perform are equivalent. It is proved by Staton in \cite[Theorem 4.1(7)]{Staton2011}; in fact, he proves a much more general statement than the one we give here.

\begin{theorem}\label{thm:T-bisimulation-span}
    Let \(T \colon \Set \to \Set\) be a functor, and let \((A, \alpha)\) be a \(T\)-coalgebra. An equivalence relation \(E \subseteq A \times A\) is a \(T\)-bisimulation equivalence on \((A,\alpha)\) if and only if, for all \((a, a') \in E\), one has \((\alpha(a), \alpha(a')) \in \lift{T}(E)\).
\end{theorem}

Using \Cref{thm:T-bisimulation-span}, we can instantiate bisimulation equivalences for F-hypergraphs, interpreted as coalgebras for the double powerset functor. We will use the fact that relation lifting respects the composition of functors in the sense of the following lemma, which can be derived from Kurz and Velebil's \cite[Proposition 2.10]{KurzVelebil}.

\begin{lemma}\label{lem:lifting-composite}
    Let \(S, T \colon \Set \to \Set\) be functors. For any relation \(E \subseteq A \times B\), we have \(\lift{ST}(E) = \lift{S}(\lift{T}(E))\).
\end{lemma}

\begin{example}[\(\power\power\)-bisimulation equivalences]
\label{eg:PP-bisimulations}
    Let \(\mathcal{H} = (A, H)\) be an F-hypergraph, and let \(\mathcal{N}_{\mathcal{H}} \colon A \to \power\power(A)\) denote the associated coalgebra. To characterize bisimulation equivalences on \(\mathcal{N}_{\mathcal{H}}\) using \Cref{thm:T-bisimulation-span}, we first need to describe \(\power\power\)-relation lifting.
    
    So, let \(E \subseteq A \times A\) be a relation. According to \Cref{lem:lifting-composite}, the relation-lifting of \(E\) by \(\power\power\) is the relation \(\lift{\power\power}(E) \subseteq \power\power (A) \times \power\power (A)\) obtained by taking \(\lift{\power}(E) \subseteq \power(A) \times \power(A)\) and lifting it by \(\power\). Using the expression for \(\lift{\power}(E)\) in \Cref{def:T-lifting} twice, one finds that \(\lift{\power\power}(E)\) is the set of pairs \((X,X') \in \power\power(A) \times \power\power(A)\) satisfying both these conditions:
    \begin{enumerate}
        \item \(\forall \: U \in X, \: \exists \: U' \in X' \text{ such that} 
            \begin{cases} 
            \forall u \in U, \: \exists \: u' \in U' \text{ with } (u,u') \in E \\
            \text{and} \\
            \forall u' \in U', \: \exists \: u \in U \text{ with } (u,u') \in E;
            \end{cases}\)
        \item \(\forall \: U' \in X', \: \exists \: U \in X \text{ such that} 
            \begin{cases} 
            \forall u' \in U', \: \exists \: u \in U \text{ with } (u,u') \in E \\
            \text{and} \\
            \forall u \in U, \: \exists \: u' \in U' \text{ with } (u,u') \in E.
            \end{cases}\)
    \end{enumerate}
    Notice that condition 1 holds for a pair \((X,X') \in \power\power(A) \times \power\power(A)\) if and only if condition 2 holds for the pair \((X',X)\).
    
    Now, suppose that \(E\) is an equivalence relation. By \Cref{thm:T-bisimulation-span}, \(E\) is a \(\power\power\)-bisimulation equivalence if and only if, for every \((a,a') \in E\), we have \((\mathcal{N}_{\mathcal{H}}(a), \mathcal{N}_{\mathcal{H}}(a')) \in \lift{\power\power} (A)\). In fact it is enough to require that, for every \((a,a') \in E\), the pair \((\mathcal{N}_{\mathcal{H}}(a), \mathcal{N}_{\mathcal{H}}(a'))\) satisfies condition 1. For, suppose that is the case; then, by the symmetry of \(E\), the pair \((\mathcal{N}_{\mathcal{H}}(a'), \mathcal{N}_{\mathcal{H}}(a))\) also satisfies condition 1, which means that \((\mathcal{N}_{\mathcal{H}}(a), \mathcal{N}_{\mathcal{H}}(a'))\) satisfies condition 2.
    
    Spelling out condition 1 for \(X = \mathcal{N}_{\mathcal{H}}(a)\) and \(X' = \mathcal{N}_{\mathcal{H}}(a')\), one sees that an equivalence relation \(E \subseteq A \times A\) is a \(\power\power\)-bisimulation equivalence if and only if it has the following property with respect to the F-hypergraph \(\mathcal{H}\): for every \((a, a') \in E\) and for each \((a, U) \in H\), there exists \((a', U') \in H\) such that
    \begin{itemize}
        \item for each \(u \in U\) there exists \(u' \in U'\) such that \((u,u') \in E\), and
        \item for each \(u' \in U'\) there exists \(u \in U\) such that \((u,u') \in E\).
    \end{itemize}
    In \Cref{sec:analysis-for-hypergraphs}, adopting the language used for graphs, we will define such an equivalence relation to be a \emph{regular equivalence} on the F-hypergraph \(\mathcal{H}\).
\end{example}


\subsection{Positional reductions for coalgebras}
\label{subsec:coalg-blockmodelling}

It is a fundamental fact of coalgebra that the forgetful functor 
\begin{align*}
    U \colon \Coalg(T) & \to \Set \\
    (A, \alpha) &\mapsto A
\end{align*}
creates all colimits \cite[Theorem 4.2]{R20}. In particular, \(\Coalg(T)\) has coequalizers---so, given a \(T\)-bisimulation equivalence \((E,\epsilon)\) on \((A,\alpha)\), we can take the coequalizer of the pair 
\(p_1,p_2 \colon (E, \epsilon) \rightrightarrows (A,\alpha)\), also known as the {quotient} of \(\alpha\) by \(\epsilon\).

\begin{definition}\label{def:coalg-blockmodel}
    Let \(T \colon \Set \to \Set\) be a functor, and \((A,\alpha)\) a \(T\)-coalgebra. Given a \(T\)-bisimulation equivalence \((E,\epsilon)\) on \((A,\alpha)\), the \define{\(T\)-blockmodel} of \(\alpha\) by \(\epsilon\) is the quotient of \(\alpha\) by \(\epsilon\) in \(\Coalg(T)\). We will denote it by \((\block{A}{E},\block{\alpha}{\epsilon})\).
\end{definition}

Rutten's description of coequalizers in \(\Coalg(T)\) provides an explicit expression for the coalgebra structure on a blockmodel \cite[\S 4.2]{R20}. Writing \(q \colon A \to \block{A}{E}\) for the quotient at the level of sets, the structure map of the blockmodel \((\block{A}{E}, \alpha/\epsilon)\) is the function
    \begin{equation}\label{eq:coalg-blockmodel-concretely}
    \begin{aligned}
        \alpha/\epsilon \colon \block{A}{E} & \to T(\block{A}{E}) \\
        [a] & \mapsto (Tq) (\alpha(a)).
    \end{aligned}
    \end{equation}
Using this expression, we can verify that \Cref{def:coalg-blockmodel} recovers the usual definition of blockmodels in the case of graphs.

\begin{example}[\(\power\)-blockmodels]
\label{eg:P-blockmodels}
    Let \(G = (A,R)\) be a graph, and \(E \subseteq A \times A\) an outward-regular equivalence. We saw in \Cref{subsec:coalg-bisimulation} that \(E\) is a bisimulation equivalence on the \(\power\)-coalgebra \(N^\mathrm{out}_G \colon A \to \power(A)\), so it carries a \(\power\)-coalgebra structure \(\epsilon \colon E \to \power(E)\). Now we will consider the \(\power\)-blockmodel of \((A, N^\mathrm{out}_G)\) by \((E,\epsilon)\) and compare it to the \(\power\)-coalgebra \(N_{\block{G}{E}}^\mathrm{out}\), where \(\block{G}{E}\) is the graph blockmodel of \(G\) by \(E\).
    
    On one hand, using expression \eqref{eq:coalg-blockmodel-concretely}, we have that
    \[
        (N^\mathrm{out}_G/\epsilon)([a]) = \power q (N^\mathrm{out}_G(a)) = \{X \in \block{A}{E} \mid \exists b \in X \text{ such that } b \in N^\mathrm{out}_G(a)\}
    \]
    for each \([a] \in \block{A}{E}\). On the other hand, by construction of \(\block{G}{E}\), we have
    \[
    N_{\block{G}{E}}^\mathrm{out}([a]) = \{X \in A / R \mid \exists b \in X \text{ such that } (a,b) \in E\}.
    \]
    Since \(N^\mathrm{out}_G(a) = \{b \mid (a,b) \in E\}\), this says exactly that 
    \[\left(\block{A}{E}, N^\mathrm{out}_G/\epsilon\right) = \left(\block{A}{E}, N_{\block{G}{E}}^\mathrm{out}\right)\]
    as coalgebras for the powerset functor. Thus, \(\power\)-blockmodels coincide with ordinary graph blockmodels with respect to outward-regular equivalences.
\end{example}

Next, we can describe blockmodels of \(\power\power\)-coalgebras, and thereby of F-hypergraphs.

\begin{example}[\(\power\power\)-blockmodels]
\label{eg:PP-blockmodels}
    Let \(\mathcal{H} = (A,H)\) be an F-hypergraph, and let \(\mathcal{N}_{\mathcal{H}} \colon A \to \power\power(A)\) denote the corresponding coalgebra for the double powerset functor. Suppose \((E,\epsilon)\) is a bisimulation equivalence on \(\mathcal{N}_{\mathcal{H}}\). Using \eqref{eq:coalg-blockmodel-concretely}, we see that the coalgebra structure of the \(\power\power\)-blockmodel \((\block{A}{E}, \block{\mathcal{N}_{\mathcal{H}}}{\epsilon})\) is given by
    \[
    (\mathcal{N}_{\mathcal{H}}/\epsilon)([a]) = (\power\power q) (\mathcal{N}_{\mathcal{H}}(a)) = \left\{ \{[u] \in \block{A}{E} \mid u \in U \} \mid U \in \mathcal{N}_{\mathcal{H}}(a) \right\}.
    \]
    for each \([a] \in \block{A}{E}\). 
    
    We can also describe the blockmodel as an F-hypergraph. Let \(\block{\mathcal{H}}{E}\) denote the F-hypergraph structure on \(\block{A}{E}\) in which \(([a],X)\) is a hyperedge if and only if there exists some \((a,U) \in H\) such that \(X = \{[u] \mid u \in U\}\).
    Then
    \[
    \mathcal{N}_{\block{\mathcal{H}}{E}}([a]) = \left\{X \subseteq \block{A}{E} \mid X = \{[u] \mid u \in U\} \text{ for some } (a,U) \in H \right\},
    \]
    and as \(\mathcal{N}_{\mathcal{H}}(a) = \{U \mid (a,U) \in H\}\), this gives \(\mathcal{N}_{\mathcal{H}} / \epsilon = \mathcal{N}_{\block{\mathcal{H}}{E}}\).
\end{example}

\Cref{def:positional-reduction-graph} introduced the term `positional reduction' for a blockmodel of a \(k\)-relational graph by an outward-regular equivalence. To generalize this, a positional reduction of a \(k\)-relational \(T\)-coalgebra should be a quotient by a bisimulation equivalence. We can be more explicit about what these look like by using the following statement, which is essentially the First Isomorphism Theorem for coalgebras, proved by Rutten in \cite[Theorem 7.1]{R20}. Since we frame it slightly differently, we will assemble the relevant statements of \cite{R20} into a brief proof.

\begin{lemma}
\label{lem:coalgebra-PR-concretely}
    For a map \(f \colon (A, (\alpha_i)_{i=1}^k) \to (B, (\beta_i)_{i=1}^k)\) in \(k\Coalg(T)\), the following are equivalent.
    \begin{itemize}
        \item There exist an equivalence relation \(E \subseteq A \times A\) and coalgebra structures \(\epsilon_1,\ldots, \epsilon_k\) on \(E\), such that \((E,\epsilon_i)\) is a bisimulation equivalence on \((A,\alpha_i)\) for each \(i\) and \(f \colon (A,\alpha_i) \to (B, \beta_i)\) is the quotient of \(\alpha_i\) by \(\epsilon_i\).
        \item The function \(f \colon A \to B\) is surjective.
    \end{itemize}
\end{lemma}

\begin{proof}
    First, suppose \(f \colon (A, \alpha_1) \to (B, \beta_1)\) is a quotient map. Then in particular it is a coequalizer, so it is an epimorphism in \(\Coalg(T)\). Rutten proves that a morphism in \(\Coalg(T)\) is epic if and only if its underlying function is surjective \cite[Proposition 4.7(1)]{R20}; hence, \(f \colon A \to B\) is surjective.

    For the converse, assume \(f \colon A \to B\) is surjective. Let \(E \subseteq A \times A\) be the equivalence relation
    \(\{(a,a') \mid f(a) = f(a')\}\);
    then \(f\) is the coequalizer in \(\Set\) of the projection functions \(p_1,p_2 \colon E \rightrightarrows A\). Since \(f \colon (A,\alpha_i) \to (B,\beta_i)\) is a coalgebra homomorphism, \(E\) is a bisimulation equivalence on \((A,\alpha_i)\): it carries a unique coalgebra structure \(\epsilon_i\) with respect to which \(p_1,p_2 \colon (E,\epsilon_i) \rightrightarrows (A,\alpha_i)\) are homomorphisms \cite[Proposition 5.7]{R20}. Since the forgetful functor creates coequalizers, \(f\) is the coequalizer of \((p_1,p_2)\) in \(\Coalg(T)\); that is, \(f \colon (A,\alpha_i) \to (B, \beta_i)\) is the quotient of \((A,\alpha_i)\) by \((E, \epsilon_i)\).
\end{proof}

On the strength of the lemma, we make the following definition.

\begin{definition}
\label{def:positional-reduction-coalgebras}
    Let \(T \colon \Set \to \Set\) be a functor. A \define{positional reduction} of a \(k\)-relational \(T\)-coalgebra \((A, (\alpha_i)_{i=1}^k)\) is a map \(f \colon (A, (\alpha_i)_{i=1}^k) \to (B, (\beta_i)_{i=1}^k)\) in \(k\Coalg(T)\) such that \(f \colon A \to B\) is surjective.
\end{definition}

Taking \(T = \power\) in \Cref{def:positional-reduction-coalgebras} recovers the definition of positional reductions for multirelational graphs (\Cref{def:positional-reduction-graph}). Taking \(T = \power\power\), it yields a notion of positional reduction for multirelational F-hypergraphs, which we will turn into a definition in \Cref{sec:analysis-for-hypergraphs}.





\section{Role analysis via Kleisli composition}\label{sec:role_kleisli}

In \Cref{sec:positional_role} we described role analysis for multirelational graphs as follows. First, given \(G = (A, (R_i)_{i=1}^k)\), one generates the \emph{semigroup of roles} \(\Role(G)\) by taking all possible composites of the graphs \((A, R_1), \ldots, (A,R_k)\) in the sense of \Cref{def:graph-composition}. The elements of the semigroup of roles are the compound relations present in the social system modelled by \(G\). A \emph{role reduction} of \(G\) is then defined to be a quotient of \(\Role(G)\): a semigroup \(\mathbb{S}\) equipped with a surjective homomorphism \(\Role(G) \twoheadrightarrow \mathbb{S}\). In \Cref{subsec:graph-functoriality} we saw that role analysis, in this sense, is functorial with respect to positional analysis: every positional reduction \(G \twoheadrightarrow H\) induces a role reduction \(\Role(G) \twoheadrightarrow \Role(H)\).

In this section we will recast each of these constructions more abstractly, in order to instantiate them beyond the setting of graphs. In \Cref{subsec:k-coalgebra-roles} we will define the \emph{semigroup of roles} in a multirelational coalgebra, and the corresponding notion of \emph{role reduction}. Then, in \Cref{subsec:functoriality}, we will prove a functoriality theorem for this abstracted form of role analysis, guaranteeing its compatibility with the form of positional analysis described in \Cref{sec:positional_bisim}.


\subsection{The semigroup of roles in a multirelational coalgebra}
\label{subsec:k-coalgebra-roles}

Our starting point is the observation---well-known to category theorists---that the binary relations between sets \(A\) and \(B\) are the morphisms from \(A\) to \(B\) in the Kleisli category, \(\Set_\power\), for the covariant powerset monad. (This terminology will be explained below.) In particular, the graph structures on a set \(A\) correspond to endomorphisms of \(A\) in \(\Set_\power\). What's more, composition in the Kleisli category is precisely composition of relations, or of graphs. It follows that for every multirelational graph \(G\), the semigroup \(\Role(G)\) can be described as a sub-semigroup of the monoid of endomorphisms of \(A\) in \(\Set_\power\).

Indeed, for this purpose it is not important that the monoid of endomorphisms of \(A\) has an identity element (see \Cref{rem:GA-monoid}). Thus, it is not important that \(\Set_\power\) is a category; nor, for that matter, that \(\power\) carries the structure of a monad. All we need is that \(\Set_\power\) is a \emph{semicategory}, and for that it is enough to know that \(\power\) is a \emph{semimonad}. This allows a degree of extra generality which will be crucial when we come to consider roles in systems represented by F-hypergraphs (\Cref{eg:double-powerset-roles}). 

Though monads are ubiquitous in category theory, the term `semimonad' is less likely to be familiar. Let us begin with the definition.

\begin{definition}\label{def:semimonad}
    A \define{semimonad} on a category \(\mathbf{C}\) is a pair \(\mathbb{T} = (T, \mu)\) consisting of a functor $T \colon \mathbf{C} \to \mathbf{C}$ and a natural transformation $\mu\colon TT \Rightarrow T$ that makes this diagram commute:
    \[
    \begin{tikzcd}
        TTT \arrow[Rightarrow, r, "T\mu"] \arrow[Rightarrow, d, "\mu_{T}", swap] & TT \arrow[Rightarrow, d, "\mu"]\\
        TT \arrow[Rightarrow, r, "\mu", swap] & T.
    \end{tikzcd}
    \]
    The natural transformation \(\mu\) is called the \define{multiplication} of \(\mathbb{T}\).
    
    A \define{monad} on \(\mathbf{C}\) is a semimonad \(\mathbb{T}\) together with a natural transformation \(\eta\) from the identity functor on \(\mathbf{C}\) to \(T\), making this diagram commute:
    \[
    \begin{tikzcd}
    T \arrow[Rightarrow, r, "\eta_{T}"] \arrow[Rightarrow, d, "T\eta", swap] \arrow[Rightarrow, rd, "\mathrm{Id}_T"] & TT \arrow[Rightarrow, d, "\mu"]\\
    TT \arrow[Rightarrow, r, swap, "\mu"] & T
    \end{tikzcd}
    \]
    The natural transformation \(\eta\) is called the \define{unit} of \(\mathbb{T}\).
\end{definition}

Just as every monad gives rise to category known as its \emph{Kleisli category} \cite[\S VI.5]{MacLane}, a semimonad gives rise to a \emph{semicategory}. A \define{semicategory} \(\mathbf{S}\) consists of a collection of objects \(\mathrm{Ob}(\mathbf{S})\); for each pair \(X,Y \in \mathrm{Ob}(\mathbf{S})\) a set of morphisms \(\mathbf{S}(X,Y)\); and for each triple \(X,Y,Z \in \mathrm{Ob}(\mathbf{S})\) a composition function
\[\mathbf{S}(Y,Z) \times \mathbf{S}(X,Y) \to \mathbf{S}(X,Z)\]
which satisfies the associative law familiar from the definition of a category. Thus, a {semicategory} is `a category without identities', just as a semigroup is `a monoid without an identity'.

\begin{definition}\label{def:kleisli-semicategory}
    Let \(\mathbb{T} = (T, \mu)\) be a semimonad on a category \(\mathbf{C}\). The \define{Kleisli semicategory} of \(\mathbb{T}\) is the semicategory \(\mathbf{C}_\mathbb{T}\) whose objects are those of \(\mathbf{C}\), and in which a morphism \(f\colon X \rightsquigarrow Y\) is a morphism \(f \colon  X \to T(Y)\) in \(\mathbf{C}\). Given \(f\colon X \rightsquigarrow Y\) and \(g\colon Y \rightsquigarrow Z\), the composite $g \ast f$ in \(\mathbf{C}_\mathbb{T}\) is defined by
    \[X \xto{f} T(Y) \xto{T(g)} TT(Z) \xto{\mu_Z} T(Z).\]
    If \(\mathbb{T}\) is a monad, meaning that it has a unit transformation \(\eta\), then \(\mathbf{C}_\mathbb{T}\) is a category, called the \define{Kleisli category} of \(\mathbb{T}\). The identity morphism \(\mathrm{Id}_X \colon X \rightsquigarrow X\) in \(\mathbf{C}_\mathbb{T}\) is the morphism \(\eta_X \colon X \to T(X)\) in \(\mathbf{C}\).
\end{definition}

For each object \(X\) in \(\mathbf{C}\), the set of endomorphisms \(\mathbf{C}_\mathbb{T}(X,X)\) is the set of morphisms from \(X\) to \(T(X)\) in \(\mathbf{C}\); that is, the set of \(T\)-coalgebra structures on \(X\). Whereas in a category the set of endomorphisms of any object forms a monoid under composition, in a semicategory this set is, \textit{a priori}, only a semigroup. But that structure is enough to let us define a \emph{semigroup of roles} and \emph{role reduction} for any multirelational \(T\)-coalgebra.

\begin{definition}\label{def:coalgebra-roles}
    Let $\mathbb{T} = (T, \mu)$ be a semimonad on $\Set$, and $\alpha = (A, (\alpha_i)_{i=1}^k)$ be a multirelational \(T\)-coalgebra. The \define{semigroup of \(\mu\)-roles in $\alpha$} is the sub-semigroup of $\Set_\mathbb{T}(A,A)$ generated by the coalgebras in \(\alpha\):
    \[\Role_\mathbb{T}(\alpha) = \langle \alpha_1, \ldots, \alpha_k \rangle_\ast.\]
    A \define{\(\mu\)-role reduction of \(\alpha\)} is a semigroup \(\mathbb{S}\) equipped with a surjective semigroup homomorphism
    \(\Role_\mathbb{T}(\alpha) \twoheadrightarrow \mathbb{S}\).
\end{definition}

Where the choice of multiplication is unambiguous, we will suppress the reference to \(\mu\).

\begin{example}[Role analysis for \(\power\)-coalgebras]
\label{eg:powerset-roles}
    The covariant powerset functor \(\power \colon \Set \to \Set\) carries the structure of a monad \(\mathbb{P} = (\power, \eta, \mu)\), whose unit \(\eta\) at a set \(A\) is the function 
    \begin{align*}
        \eta_A \colon A & \to \power(A) \\
        a &\mapsto \{a\}
    \end{align*}
    and whose multiplication \(\mu\) at \(A\) takes each set of subsets to their union:
    \begin{align*}
        \mu_A \colon \power \power(A) & \to \power (A) \\
        \mathcal{U} & \mapsto \bigcup_{V \in \mathcal{U}} V.
    \end{align*}
    Since coalgebras for the powerset functor are graphs, the elements of \(\Set_\mathbb{P}(A,A)\) are in bijection with the graph structures on \(A\). Moreover, composition in the Kleisli category corresponds, under this bijection, to the composition of graphs as defined in \Cref{def:graph-composition}. Thus, for each set \(A\), the semigroup  \(\Set_\mathbb{P}(A,A)\) is isomorphic to the semigroup \(\mathbb{G}(A)\) of graphs on \(A\), and given any multirelational graph \(G = (A, (R_i)_{i=1}^k)\), we have
    \[\Role_\mathbb{P}(N^\mathsf{out}_G) \cong \Role(G).\]
    It follows that a role reduction of \(N^\mathsf{out}_G\), in the sense of \Cref{def:coalgebra-roles}, is a role reduction of \(G\) in the sense of \Cref{def:role-reduction}.
\end{example}

\Cref{eg:powerset-roles} confirms that role analysis for \(\power\)-coalgebras recovers the established notion of role analysis for graphs. The next example makes use of the category-theoretic framework to define not just one, but two notions of role analysis for F-hypergraphs.

\begin{example}
[Role analysis for \(\power\power\)-coalgebras]
\label{eg:double-powerset-roles}
    The double covariant powerset functor \(\power\power\colon \Set \to \Set\) does not admit the structure of a monad: Klin and Salamanca have shown that it cannot be equipped with a unit transformation \cite[Theorem 3.2]{KLIN2018261}. However, \(\power\power\) \emph{does} admit the structure of a semimonad; in fact, it carries at least two different multiplications satisfying \Cref{def:semimonad}.
    
    Let \(\mu \colon \power\power \Rightarrow \power\) be the multiplication of the covariant powerset monad, and let $\nu_1, \nu_2 \colon \mathcal{PPPP} \Rightarrow \mathcal{PP}$ be the natural transformations defined by
    \[\nu_1 = \mu_\power \circ \mu_\mathcal{PP}.\]
    and 
    \[\nu_2 = \power\mu \circ \mathcal{PP}\mu\]
    In a discussion on the mathematical blog \emph{The \(n\)-Category Caf\'e}, Baez and Egan proved that both $\mathbb{PP}_1 = (\mathcal{PP}, \nu_1)$ and $\mathbb{PP}_2 = (\mathcal{PP}, \nu_2)$ are semimonads \cite{BaezEganncat}.
    
    Since coalgebras for the double powerset functor are  F-hypergraphs, both \(\Set_{\mathbb{PP}_1}(A,A)\) and  \(\Set_{\mathbb{PP}_2}(A,A)\) are in bijection with the set of F-hypergraph structures on \(A\). Their composition operations, though, are different; we will denote them by \(\ast_1\) and \(\ast_2\) respectively. This means that, given a multirelational F-hypergraph \(\mathcal{H} = (A, (\eta_i)_{i=1}^k)\), we can generate two distinct semigroups of roles in \(\mathcal{H}\), namely 
    \[\Role_{\mathbb{PP}_1}(\mathcal{H}) = \langle \eta_1, \ldots, \eta_k \rangle_{\ast_1} \subseteq \Set_{\mathbb{PP}_1}(A,A)\]
    and
    \[\Role_{\mathbb{PP}_2}(\mathcal{H}) = \langle \eta_1, \ldots, \eta_k \rangle_{\ast_2} \subseteq \Set_{\mathbb{PP}_2}(A,A)\] 
    Thus, the two semimonad structures on \(\power\power\) give rise to two distinct forms of role analysis for multirelational F-hypergraphs, both of which will be described explicitly in \Cref{sec:analysis-for-hypergraphs}.
\end{example}


\subsection{The functoriality theorem}
\label{subsec:functoriality}

In \Cref{sec:positional_bisim}, we abstracted the main ideas and constructions of positional analysis from graphs to \(T\)-coalgebras for an arbitrary functor \(T\colon \Set \to \Set\). In \Cref{subsec:k-coalgebra-roles} we showed that if \(T\) happens to carry the structure of a semimonad, then one can define role analysis for its coalgebras, too. Now we are going to prove that role analysis, in this very general sense, is compatible with positional analysis. Our main result, \Cref{thm:functoriality-coalgebras}, specializes to the functoriality theorem for graphs (\Cref{thm:functoriality-graphs}) by taking \(T\) to be the powerset functor \(\power\).

Given a semimonad \(\mathbb{T} = (T, \mu)\) on \(\Set\), let \(k\Coalg(T)_\mathsf{PR}\) denote the category of \(k\)-relational \(T\)-coalgebras and positional reductions; that is, the subcategory of \(k\Coalg(T)\) containing all the objects, but only those morphisms whose underlying functions are surjective. As before, let \(\mathbf{SGrp}_\mathsf{Q}\) denote the category of semigroups and surjective semigroup homomorphisms.

\begin{theorem}
\label{thm:functoriality-coalgebras}
    Let \(\mathbb{T} = (T, \mu)\) be a semimonad on \(\Set\). The assignment of the semigroup of roles extends to a functor
    \[\Role_\mathbb{T} \colon k\Coalg(T)_\mathsf{PR} \to \mathbf{SGrp}_\mathsf{Q}.\]
    That is, \(\mu\)-role analysis is functorial with respect to \(T\)-positional reductions.
\end{theorem}

Before we prove the theorem, let us elaborate on the statement. In the domain category, an object is a set \(A\) equipped with a tuple \((\alpha_1,\ldots, \alpha_k)\) of \(T\)-coalgebra structures. As \(\mathbb{T}\) is a semimonad, each of these coalgebra structures can be regarded as an endomorphism of \(A\) in the Kleisli semicategory \(\Set_\mathbb{T}\), and \(\Role_\mathbb{T}\) is defined on objects by taking the set \(\{\alpha_1,\ldots, \alpha_k\}\) to the semigroup it generates under the composition operation \(\ast\) in \(\Set_\mathbb{T}(A,A)\). 

The claim is that every surjective function \(f \colon A \to B\) which is a homomorphism from \(\alpha = (A, (\alpha_i)_{i=1}^k)\) to \(\beta = (B, (\beta_i)_{i=1}^k)\) induces a surjective semigroup homomorphism \(\Role_\mathbb{T}(f) \colon \Role_\mathbb{T}(\alpha) \to \Role_\mathbb{T}(\beta)\), in such a way that
\begin{equation}\label{eq:functoriality-conditions}
    \Role_\mathbb{T}(g \circ f) = \Role_\mathbb{T}(g) \circ \Role_\mathbb{T}(f) \: \text{ and } \: \Role_\mathbb{T}(\mathrm{Id}_\alpha) = \mathrm{Id}_{\Role_\mathbb{T}(\alpha)}. 
\end{equation}
In fact we will prove that, provided there exists \emph{at least one} morphism from \(\alpha\) to \(\beta\) in \(k\Coalg(T)_\mathsf{PR}\), the function
\begin{equation}\label{eq:semigp-hom-def}
    \begin{aligned}
    \{\alpha_1, \ldots, \alpha_k\} &\to \{\beta_1, \ldots, \beta_k\} \\
    \alpha_i &\mapsto \beta_i
    \end{aligned}
\end{equation}
extends to a surjective semigroup homomorphism \(\Role_\mathbb{T}(\alpha) \to \Role_\mathbb{T}(\beta)\). We will declare this homomorphism to be \(\Role_\mathbb{T}(f)\) for \emph{every} \(f \colon \alpha \to \beta\). 

The main task is to show that the mapping in \eqref{eq:semigp-hom-def} really does extend to a homomorphism; if so, then the functoriality conditions in \eqref{eq:functoriality-conditions} are plainly satisfied. The proof makes use of two lemmas.

\begin{lemma}\label{lem:f_mult}
    Let \(\alpha_1, \alpha_2\) be \(T\)-coalgebras on a set \(A\) and let \(\beta_1, \beta_2\) be \(T\)-coalgebras on a set \(B\). Suppose \(f \colon A \to B\) is a function which is a \(T\)-coalgebra homomorphism from \(\alpha_1\) to  \(\beta_1\) and from \(\alpha_2\) to \(\beta_2\). Then \(f\) is a \(T\)-coalgebra homomorphism from \(\alpha_1 \ast \alpha_2\) to \(\beta_1 \ast \beta_2\).
\end{lemma}

\begin{proof}
Consider the diagram
    \[ 
    \begin{tikzcd}
        A \arrow[bend left=25]{rrr}{\alpha_1 \ast \alpha_2} \arrow{r}{\alpha_1} \arrow{d}{f} & TA \arrow{r}{T\alpha_2} \arrow{d}{Tf} & TTA \arrow{r}{\mu_A} \arrow{d}{TTf} & TA \arrow{d}{Tf} \\
        B \arrow[bend right=25, swap]{rrr}{\beta_1 \ast \beta_2} \arrow[swap]{r}{\beta_1} & TB \arrow[swap]{r}{T\beta_2} & TTB \arrow[swap]{r}{\mu_B} & TB.
    \end{tikzcd}
    \]
The squares on the left and in the middle commute because \(f\) is a \(T\)-coalgebra homomorphism; the square on the right commutes by the naturality of \(\mu\). So the outer diagram commutes, which says that \(f\) is a \(T\)-coalgebra homomorphism from \(\alpha_1 \ast \alpha_2\) to \(\beta_1 \ast \beta_2\).
\end{proof}

\begin{lemma}\label{lem:prod_eq}
    Let \(\alpha\) be a \(T\)-coalgebra on a set \(A\) and let \(\beta_1, \beta_2\) be \(T\)-coalgebras on a set \(B\). Suppose there exists a surjective function $f \colon  A \to B$ which is a \(T\)-coalgebra homomorphism  \(\alpha \to \beta_1\) and \(\alpha \to \beta_2\). Then $\beta_1 = \beta_2$. 
\end{lemma}

\begin{proof}
As $f$ is surjective, it suffices to show that $\beta_1 \circ f = \beta_2 \circ f$. And indeed,
\[\beta_1 \circ f = Tf \circ \alpha = \beta_2 \circ f\]
where the first equation expresses that \(f\) is a homomorphism from \(\alpha\) to \(\beta_1\) and the second that it is a homomorphism from \(\alpha\) to \(\beta_2\).
\end{proof}

With these, we can prove the theorem.

\begin{proof}[Proof of \Cref{thm:functoriality-coalgebras}]
    Let \(\alpha = (A, (\alpha_i)_{i=1}^k)\) and \(\beta = (B, (\beta_i)_{i=1}^k)\) be \(k\)-relational \(T\)-coalgebras, and suppose there exists a surjection \(f \colon  A \to B\) which is a \(T\)-coalgebra homomorphism \(\alpha_i \to \beta_i\) for each \(i \in \{1,\ldots,k\}\). Our goal is to prove that the function
    \begin{align*}
    \{\alpha_1, \ldots, \alpha_k\} &\to \{\beta_1, \ldots, \beta_k\} \\
    \alpha_i &\mapsto \beta_i
    \end{align*}
    extends to a surjective semigroup homomorphism \(\Role_\mathbb{T}(\alpha) \to \Role_\mathbb{T}(\beta)\).

    Every element \(\rho \in \Role_\mathbb{T}(\alpha)\) can be written in at least one way as a composite of the generating coalgebras: say, \(\rho = \alpha_{i_1} \ast \cdots \ast \alpha_{i_m}\) for \(i_1, \ldots, i_m \in \{1,\ldots,k\}\). We would like to specify
    \[h\colon \Role_\mathbb{T}(\alpha) \to \Role_\mathbb{T}(\beta)\]
    by \(h(\rho) = \beta_{i_1} \ast \cdots \ast \beta_{i_m}\). If \(h\) is well-defined, it is evidently a semigroup homomorphism and surjective; what we need to check is that \(h(\rho)\) does not depend on the chosen decomposition of \(\rho\) into generators.

    So, suppose \(\alpha_{i_1} \ast \cdots \ast \alpha_{i_m} = \alpha_{j_1} \ast \cdots \ast \alpha_{j_n}\) for some \(j_1, \ldots, j_n \in \{1,\ldots, k\}\). We would like to see that \(\beta_{i_1} \ast \cdots \ast \beta_{i_m} = \beta_{j_1} \ast \cdots \ast \beta_{j_n}\). Since the surjective function \(f \colon A \to B\) is a coalgebra homomorphism \(\alpha_{i_\ell} \to \beta_{i_p}\) for each \(p \in \{1, \ldots, m\}\), and also a coalgebra homomorphism \(\alpha_{j_q} \to \beta_{j_q}\) for each \(q \in \{1, \ldots, n\}\), applying \Cref{lem:f_mult} repeatedly tells us that it is a coalgebra homomorphism
    \[\alpha_{i_1} \ast \cdots \ast \alpha_{i_m} \to \beta_{i_1} \ast \cdots \ast \beta_{i_m}\]
    and also a coalgebra homomorphism
    \[\alpha_{j_1} \ast \cdots \ast \alpha_{j_n} \to \beta_{j_1} \ast \cdots \ast \beta_{j_n}.\]
    Given that \(\alpha_{i_1} \ast \cdots \ast \alpha_{i_m} = \alpha_{j_1} \ast \cdots \ast \alpha_{j_n}\), \Cref{lem:prod_eq} now tells us that
    \[\beta_{i_1} \ast \cdots \ast \beta_{i_m} = \beta_{j_1} \ast \cdots \ast \beta_{j_n},\]
    and hence that \(h\colon \Role_\mathbb{T}(\alpha) \to \Role_\mathbb{T}(\beta)\) is well-defined.

    Now, for each \(f \colon \alpha \to \beta\) in \(k\Coalg(T)_\mathsf{PR}\), we declare that \(\Role_\mathbb{T}(f) = h\). Then \(\Role_\mathbb{T}(\mathrm{Id}_\alpha) = \mathrm{Id}_{\Role_\mathbb{T}(\alpha)}\) and \(\Role_\mathbb{T}(g \circ f) = \Role_\mathbb{T}(g) \circ \Role_\mathbb{T}(f)\), so \(\Role_\mathbb{T}\) becomes a functor, as claimed.
\end{proof}

\begin{remark}
    Though we state \Cref{thm:functoriality-coalgebras} for coalgebras over $\Set$, the result is really far more general. If we instead worked with a semimonad \(\mathbb{T}\) on an arbitrary category $\mathbf{C}$, and defined the category of $k$-relational $T$-coalgebras using epimorphisms in \(\mathbf{C}\) in place of surjections, the same proof would go through unchanged. Indeed, the proof never uses any special property of $\Set$; even \Cref{lem:prod_eq} relies only on the fact that surjections are epimorphisms in $\Set$. We present the theorem (as with the rest of the paper) in $\Set$ for accessibility, and because this level of generality suffices for our purposes.
    
    What does change outside $\Set$ is that the relationship between epimorphic coalgebra morphisms and bisimulation equivalences becomes subtler; see \cite{Mavoungou_2019}. Thus, to develop positional analysis for coalgebras over a functor \(T \colon \mathbf{C} \to \mathbf{C}\) would require additional assumptions on both \(\mathbf{C}\) and $T$.
\end{remark}


\section{Positional and role analysis on hypergraphs}
\label{sec:analysis-for-hypergraphs}

Finally, we are equipped to describe positional and role analysis for social systems modelled by hypergraphs. The definitions in this section are given in terms of F-hypergraphs, but each one can be specialized to undirected hypergraphs using \Cref{prop:hypergraphs-F-hypergraphs}.

We introduce positional analysis first, in \Cref{subsec:hypergraph-positional-analysis}, where we define \emph{blockmodels}, \emph{regular equivalences}, and \emph{positional reductions} of F-hypergraphs and give an example.  In \Cref{subsec:hypergraph-role-analysis} we describe role analysis. In contrast with the case for graphs, there are two distinct notions of role analysis for multirelational F-hypergraphs that emerge naturally from our framework. We refer to these as producing a semigroup of \emph{tight roles} and a semigroup of \emph{loose roles}, and we give two examples to illustrate how they differ. 

In \Cref{subsec:hypergraph-functoriality} we apply our main result, \Cref{thm:functoriality-coalgebras}, to prove a functoriality theorem expressing that both tight and loose role analysis are compatible with positional analysis for hypergraphs. In particular, every positional reduction of hypergraphs induces a role reduction of the semigroup of tight roles \emph{and} the semigroup of loose roles.


\subsection{Positional analysis}
\label{subsec:hypergraph-positional-analysis}

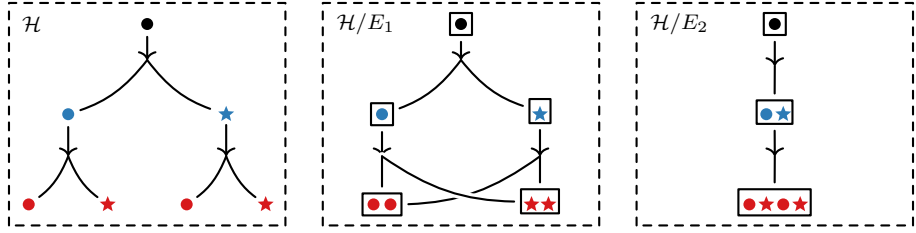
\begin{figure}
    \centering
    \begin{tikzpicture}[line cap=round, line join=round, thick, yscale=0.8, xscale=0.75] 
    \tikzset{
    dot/.style={ circle, fill, inner sep=1.6pt},
    stardot/.style={ star, star points=5, star point ratio=2.25, fill, draw=none, minimum size=5.5pt, inner sep=1pt}, 
    }
    \def\bend{16}
    \def\gap{2.2pt} 

    \node[dot] (r) at ( 0, 2.2) {}; 
    \node[dot, blue] (a) at (-1.4, 0.7) {}; 
    \node[stardot, blue] (b) at ( 1.4, 0.7) {}; 
    \node[dot, red] (a1) at (-2.1,-0.8) {}; 
    \node[stardot, red] (a2) at (-0.7,-0.8) {}; 
    \node[dot, red] (b1) at ( 0.7,-0.8) {}; 
    \node[stardot, red] (b2) at ( 2.1,-0.8) {};

    \coordinate (sr) at ( 0, 1.6); 
    \coordinate (sa) at (-1.4, 0); 
    \coordinate (sb) at ( 1.4, 0);

    \draw[->, shorten <=\gap] (r) -- (sr); 
    \draw[->, shorten <=\gap] (a) -- (sa); 
    \draw[->, shorten <=\gap] (b) -- (sb);

    \draw[shorten >=\gap] (sr) to[bend left=\bend] (a); 
    \draw[shorten >=\gap] (sr) to[bend right=\bend] (b);

    \draw[shorten >=\gap] (sa) to[bend left=\bend] (a1); 
    \draw[shorten >=\gap] (sa) to[bend right=\bend] (a2);

    \draw[shorten >=\gap] (sb) to[bend left=\bend] (b1); 
    \draw[shorten >=\gap] (sb) to[bend right=\bend] (b2); 
    
    \draw[dashed] (-2.45,-1.15) rectangle (2.45,2.55); 
    \node[anchor=east] (g) at ( -1.7, 2.2) {\footnotesize {\(\mathcal{H}\)}};
    \end{tikzpicture}
    \quad
    \begin{tikzpicture}[line cap=round, line join=round, thick, yscale=0.8, xscale=0.75] 
    \tikzset{dot/.style={circle, fill, inner sep=1.6pt}, 
    stardot/.style={star, star points=5, star point ratio=2.25, fill, draw=none, minimum size=5.5pt, inner sep=1pt},
    over/.style={preaction={draw=white, line width=\pgflinewidth+1.2pt}}, } 
    \def\bend{16}
    \def\gap{2.2pt} 

    \node[dot] (r) at ( 0, 2.2) {}; 
    \node[rectangle,draw,fit=(r),inner sep=1.5] (rr) {};
    \node[dot, blue] (a) at (-1.4, 0.7) {}; 
    \node[stardot, blue] (b) at ( 1.4, 0.7) {}; 
    \node[rectangle,draw,fit=(a),inner sep=1.5] (aa) {};
    \node[rectangle,draw,fit=(b),inner sep=1.5] (bb) {};
    \node[dot, red] (a1) at (-1.55,-0.8) {}; 
    \node[dot, red] (b1) at (-1.25,-0.8) {}; 
    \node[rectangle,draw,fit=(a1) (b1),inner sep=1.5] (ab1) {};
    \node[stardot, red] (a2) at ( 1.25,-0.8) {};
    \node[stardot, red] (b2) at ( 1.55,-0.8) {};
    \node[rectangle,draw,fit=(a2) (b2),inner sep=1.5] (ab2) {};

    \coordinate (sr) at ( 0, 1.6); 
    \coordinate (sa) at (-1.4, 0); 
    \coordinate (sb) at ( 1.4, 0);

    \draw[->, shorten <=\gap] (rr) -- (sr); 
    \draw[->, shorten <=\gap] (aa) -- (sa); 
    \draw[->, shorten <=\gap] (bb) -- (sb);

    \draw[shorten >=\gap] (sr) to[bend left=\bend] (aa); 
    \draw[shorten >=\gap] (sr) to[bend right=\bend] (bb);
    
    \draw[shorten >=\gap] (sb) to[bend left=\bend] (ab1); 
    
    \draw[shorten >=\gap, over] (sa) to[bend right=\bend] (ab2);
    \draw[shorten >=\gap] (sa) -- (ab1); 
    \draw[shorten >=\gap] (sb) -- (ab2); 

    \draw[dashed] (-2.45,-1.15) rectangle (2.45,2.55); 
    \node[anchor=east] (g) at ( -1, 2.2) {\footnotesize {\(\block{\mathcal{H}}{E_1}\)}};
    \end{tikzpicture}
    \quad
    \begin{tikzpicture}[line cap=round, line join=round, thick, yscale=0.8, xscale=0.75] 
    \tikzset{
    dot/.style={ circle, fill, inner sep=1.6pt},
    stardot/.style={ star, star points=5, star point ratio=2.25, fill, draw=none, minimum size=5.5pt, inner sep=1pt}, 
    } 
    \def\bend{16}
    \def\gap{2.2pt} 

    \node[dot] (r) at ( 0, 2.2) {}; 
    \node[rectangle,draw,fit=(r),inner sep=1.5] (rr) {};
    \node[dot, blue] (a) at (-.15, 0.7) {}; 
    \node[stardot, blue] (b) at ( .15, 0.7) {}; 
    \node[rectangle,draw,fit=(a) (b),inner sep=1] (ab) {};
    \node[dot, red] (a1) at (-.45,-0.8) {}; 
    \node[stardot, red] (a2) at (-.15,-0.8) {}; 
    \node[dot, red] (b1) at ( .15,-0.8) {}; 
    \node[stardot, red] (b2) at ( .45,-0.8) {}; 
    \node[rectangle,draw,fit=(a1) (a2) (b1) (b2),inner sep=1.5] (b12) {};

    \coordinate (sr) at ( 0, 1.5); 
    \coordinate (sa) at ( 0, 0); 

    \draw[->, shorten <=\gap] (rr) -- (sr); 
    \draw[->, shorten <=\gap] (ab) -- (sa); 

    \draw[shorten >=\gap] (sr) -- (ab); 
    \draw[shorten >=\gap] (sa) -- (b12); 
    
    \draw[dashed] (-2.45,-1.15) rectangle (2.45,2.55); 
    \node[anchor=east] (g) at ( -1, 2.2) {\footnotesize {\(\block{\mathcal{H}}{E_2}\)}};
    \end{tikzpicture}

    \caption{An F-hypergraph and two of its blockmodels. See \Cref{eg:hypergraph-blockmodels}.}
    \label{fig:hyper-block-more-info}
\end{figure}

Just as in the case of graphs, \emph{any} equivalence relation on the set of vertices of an F-hypergraph \(\mathcal{H}\) can be used to produce a new, smaller F-hypergraph whose structure reflects that of \(\mathcal{H}\). Echoing the terminology for graphs, we will call this a \emph{blockmodel}.

\begin{definition}
    Let \(\mathcal{H} = (A, H)\) be an F-hypergraph. The \define{blockmodel} of \(\mathcal{H}\) with respect to \(E\) is the F-hypergraph \(\block{\mathcal{H}}{E} = (\block{A}{E}, \block{H}{E})\) in which \(([a],X)\) is a hyperedge if and only if there exists \((a,U) \in H\) such that \(X = \{[u] \mid u \in U\}\).
\end{definition}

Just as in the case of graphs, some equivalence relations produce blockmodels that are nicer---more sensible---than others. The theory introduced in Sections \ref{sec:ss_coalg} and \ref{sec:positional_bisim} tells us how these equivalence relations should be characterized. Echoing the terminology for graphs, we will call them \emph{regular equivalences}.

\begin{definition}\label{def:regular-equiv-hypergraph}
    Let \(\mathcal{H} = (A, H)\) be an F-hypergraph. A \define{regular equivalence} on \(\mathcal{H}\) is an equivalence relation \(E \subseteq A \times A\) with the following property. For every \((a, a') \in E\) and for each \((a, U) \in H\), there exists \((a', U') \in H\) such that
    \begin{itemize}
        \item for each \(u \in U\) there exists \(u' \in U'\) such that \((u,u') \in E\), and
        \item for each \(u' \in U'\) there exists \(u \in U\) such that \((u,u') \in E\).
    \end{itemize}
\end{definition}

\begin{example}\label{eg:hypergraph-blockmodels}
    Consider the F-hypergraph \(\mathcal{H} = (A,H)\) on the left in \Cref{fig:hyper-block-more-info}. Suppose the elements of \(A\) are the members of a family and \((a,U)\) is a hyperedge in \(H\) if and only if the members of \(U\) are the parents of \(A\). \Cref{fig:hyper-block-more-info} shows blockmodels of \(\mathcal{H}\) by two different regular equivalences. The three blocks in \(\block{\mathcal{H}}{E_2}\) correspond to the `generations' in this family. The five blocks in \(\block{\mathcal{H}}{E_1}\) might, for instance, correspond to the positions of `child', `mother', `father', `grandmother' and `grandfather'.
\end{example}

Next, echoing \Cref{def:positional-reduction-graph} for graphs, we will define \emph{positional reductions} of multirelational F-hypergraphs.

\begin{definition}
\label{def:blockmodel-reduction-hypergraph}
    A \define{positional reduction} of a \(k\)-relational F-hypergraph \(\mathcal{H}\) is a map of \(k\)-relational F-hypergraphs \(\mathcal{H} \to \mathcal{K}\) which is surjective on vertices and reflects hyperedges.
\end{definition}

The idea is that a positional reduction is a map that reduces a multirelational F-hypergraph to its blockmodel by a regular equivalence. The next lemma confirms that \Cref{def:blockmodel-reduction-hypergraph} correctly characterizes such maps.

\begin{lemma}\label{lem:pos-red-equiv}
    Let \(\mathcal{H} = (A, (H_i)_{i=1}^k)\) and \(\mathcal{K} = (B, (K_i)_{i=1}^k)\) be k-relational F-hypergraphs and \(f \colon \mathcal{H}\to \mathcal{K}\) a map in \(k\FHGraph\). The following are equivalent.
    \begin{itemize}
        \item The map \(f\) is a positional reduction.
        \item The equivalence relation \(E = \{(a,a') \mid f(a) = f(a')\}\) is a regular equivalence on \((A,H_i)\) for each \(i \in \{1, \ldots, k\}\), and \((B,K_i) = (\block{A}{E},\block{H_i}{E})\).
    \end{itemize}
\end{lemma}

\begin{proof}
    By the proof of \Cref{prop:CoalgPP-HGraph-functor}, a map in \(k\FHGraph\) is a map of \(k\)-relational \(\power\power\)-coalgebras if and only if it reflects hyperedges. By \Cref{lem:coalgebra-PR-concretely}, a map of \(k\)-relational \(\power\power\)-coalgebras is a quotient by a bisimulation equivalence if and only if it is surjective on vertices. And by \Cref{eg:PP-bisimulations}, a regular equivalence on \((A,H_i)\) is precisely a bisimulation equivalence on the corresponding \(\power\power\)-coalgebra.
\end{proof}

In \Cref{subsec:hypergraph-role-analysis} we will define \emph{role reductions} for multirelational F-hypergraphs, and in \Cref{subsec:hypergraph-functoriality} we will prove that every positional reduction induces a role reduction.


\subsection{Role analysis}
\label{subsec:hypergraph-role-analysis}

In the case of a multirelational graph \((A, (R_i)_{i=1}^k)\), there is just one sensible definition for the semigroup of roles: it has to be the semigroup generated by \(\{R_1,\ldots, R_k\}\) under composition of relations. For hypergraphs, things are not quite so straightforward. Given two F-hypergraphs on the same set, there is more than one way to form a `composite', and these give rise to at least two distinct semigroups of roles.

To see this, consider the 2-relational F-hypergraph \((A, (H_1,H_2))\) on the left in \Cref{fig:two-hypergraph-composites}. There are two equally natural ways to concatenate the hyperedges in \(H_1\) and \(H_2\) to obtain a new F-hypergraph \(H_2 \circ H_1\) on \(A\). Given the unique hyperedge \((a,V)\) in \(H_1\), the first approach treats each vertex in \(V\) separately, adding a hyperedge \((a,U)\) to \(H_2 \circ H_1\) for every \(b \in V\) and every hyperedge \((b,U)\) in \(H_2\). The second approach does not distinguish the vertices in \(V\), but instead adds a single hyperedge to \(H_2 \circ H_1\) whose source is \(a\) and whose target is the union over all sets \(U\) that admit a hyperedge from \(b\) in \(H_2\). We will refer to these approaches as producing the \emph{tight} and \emph{loose} composites, respectively.

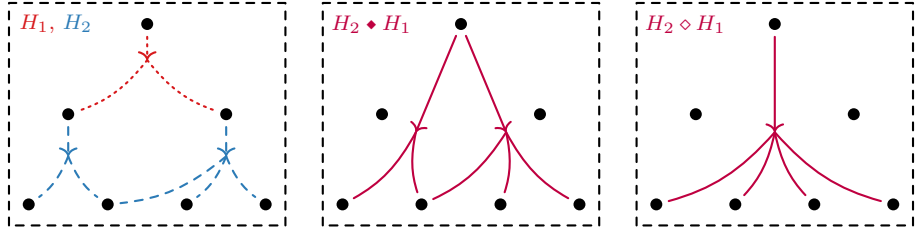
\begin{figure}
    \centering
    \begin{tikzpicture}[line cap=round, line join=round, thick, yscale=0.8, xscale=0.75] 
    \tikzset{ dot/.style={circle, fill, inner sep=1.6pt}, } 
    \def\bend{16}
    \def\gap{2.2pt} 

    \node[dot] (r) at ( 0, 2.2) {}; 
    \node[dot] (a) at (-1.4, 0.7) {}; 
    \node[dot] (b) at ( 1.4, 0.7) {}; 
    \node[dot] (a1) at (-2.1,-0.8) {}; 
    \node[dot] (a2) at (-0.7,-0.8) {}; 
    \node[dot] (b1) at ( 0.7,-0.8) {}; 
    \node[dot] (b2) at ( 2.1,-0.8) {};

    \coordinate (sr) at ( 0, 1.6); 
    \coordinate (sa) at (-1.4, 0); 
    \coordinate (sb) at ( 1.4, 0);

    \draw[->, shorten <=\gap, red, dotted] (r) -- (sr); 
    \draw[->, shorten <=\gap, blue, dashed] (a) -- (sa); 
    \draw[->, shorten <=\gap, blue, dashed] (b) -- (sb);

    \draw[shorten >=\gap, red, dotted] (sr) to[bend left=\bend] (a); \draw[shorten >=\gap, red, dotted] (sr) to[bend right=\bend] (b);

    \draw[shorten >=\gap, blue, dashed] (sa) to[bend left=\bend] (a1);
    \draw[shorten >=\gap, blue, dashed] (sa) to[bend right=\bend] (a2);

    \draw[shorten >=\gap, blue, dashed] (sb) to[bend left=\bend] (a2);
    \draw[shorten >=\gap, blue, dashed] (sb) to[bend left=\bend] (b1);
    \draw[shorten >=\gap, blue, dashed] (sb) to[bend right=\bend] (b2); 
    
    \draw[dashed] (-2.45,-1.15) rectangle (2.45,2.55); 
    \node[anchor=east] (g) at ( -.8, 2.2) {\footnotesize {\color{red}{$H_1$}, \color{blue}{$H_2$}}}; 
    \end{tikzpicture}
    \quad
    \begin{tikzpicture}[line cap=round, line join=round, thick, yscale=0.8, xscale=0.75]
    \tikzset{ dot/.style={circle, fill, inner sep=1.6pt}, } 
    \def\bend{16}
    \def\gap{2.2pt} 

    \node[dot] (r) at ( 0, 2.2) {}; 
    \node[dot] (a) at (-1.4, 0.7) {}; 
    \node[dot] (b) at ( 1.4, 0.7) {}; 
    \node[dot] (a1) at (-2.1,-0.8) {}; 
    \node[dot] (a2) at (-0.7,-0.8) {}; 
    \node[dot] (b1) at ( 0.7,-0.8) {}; 
    \node[dot] (b2) at ( 2.1,-0.8) {};

    \coordinate (sa) at (-.8, 0.4); 
    \coordinate (sb) at ( .8, 0.4);

    \draw[->, shorten <=\gap, purple] (r) -- (sa); 
    \draw[->, shorten <=\gap, purple] (r) -- (sb);

    \draw[shorten >=\gap, purple] (sa) to[bend left=\bend] (a1);
    \draw[shorten >=\gap, purple] (sa) to[bend right=\bend] (a2);

    \draw[shorten >=\gap, purple] (sb) to[bend left=\bend] (a2);
    \draw[shorten >=\gap, purple] (sb) to[bend left=\bend] (b1);
    \draw[shorten >=\gap, purple] (sb) to[bend right=\bend] (b2); 
    
    \draw[dashed] (-2.45,-1.15) rectangle (2.45,2.55); 
    \node[anchor=east] (g) at ( -.7, 2.2) {\footnotesize {\color{purple}{$H_2 \blackdiamond H_1$}}}; 
    \end{tikzpicture}
    \quad
    \begin{tikzpicture}[line cap=round, line join=round, thick, yscale=0.8, xscale=0.75]
    \tikzset{ dot/.style={circle, fill, inner sep=1.6pt}, } 
    \def\bend{16}
    \def\gap{2.2pt} 

    \node[dot] (r) at ( 0, 2.2) {}; 
    \node[dot] (a) at (-1.4, 0.7) {}; 
    \node[dot] (b) at ( 1.4, 0.7) {}; 
    \node[dot] (a1) at (-2.1,-0.8) {}; 
    \node[dot] (a2) at (-0.7,-0.8) {}; 
    \node[dot] (b1) at ( 0.7,-0.8) {}; 
    \node[dot] (b2) at ( 2.1,-0.8) {};

    \coordinate (s) at ( 0, 0.4);

    \draw[->, shorten <=\gap, purple] (r) -- (s);

    \draw[shorten >=\gap, purple] (s) to[bend left=\bend] (a1); \draw[shorten >=\gap, purple] (s) to[bend left=\bend] (a2);

    \draw[shorten >=\gap, purple] (s) to[bend right=\bend] (b1); \draw[shorten >=\gap, purple] (s) to[bend right=\bend] (b2); 
    
    \draw[dashed] (-2.45,-1.15) rectangle (2.45,2.55); 
    
    \node[anchor=east] (g) at ( -.7, 2.2) {\footnotesize {\color{purple}{$H_2 \diamond H_1$}}}; 
    
    \end{tikzpicture}

    \caption{On the left is a multirelational F-hypergraph \(\mathcal{H} = (A, (H_1,H_2))\). To its right are the tight composite \(\mathcal{H}_2 \blackdiamond \mathcal{H}_1\) and the loose composite \(\mathcal{H}_2 \diamond \mathcal{H}_1\).}
    \label{fig:two-hypergraph-composites}
\end{figure}

\begin{definition}
    Given two F-hypergraphs \(\mathcal{H}_1 = (A, H_1)\) and \(\mathcal{H}_2 = (A, H_2)\), their \define{tight composite} is the F-hypergraph \(\mathcal{H}_2 \blackdiamond \mathcal{H}_1 = (A, H_2 \blackdiamond H_1)\) in which
    \[H_2 \blackdiamond H_1 = \left\{(a,U) \: \middle| \: \text{there exist } (a,V) \in H_1 \text{ and } b \in V \text{ such that } (b,U) \in H_2 \right\}.\]
    Their \define{loose composite} is the F-hypergraph \(\mathcal{H}_2 \diamond \mathcal{H}_1 = (A, H_2 \diamond H_1)\) in which
    \[H_2 \diamond H_1 = \left\{(a,W) \mid \text{there exists } (a,V) \in H_1 \text{ such that } W = \bigcup_{b \in V} \bigcup_{(b,U) \in H_2} U \right\}.\]
\end{definition}

The next example suggests a semantic interpretation for the information captured by the tight and loose composites.

\begin{example}[Tight and loose composition]
\label{eg:hypergraph-composites}
    Consider a `family tree' F-hypergraph \(\mathcal{H} = (A, H)\), in which the elements of \(A\) are the members of a family and \((a,U)\) is a hyperedge if and only if the members of \(U\) are the parents of \(A\).
    In this F-hypergraph, the neighbourhood of each \(a \in A\) contains a single element:
    \[\mathcal{N}_\mathcal{H}(a) = \{\{\text{parents of \(a\)}\}\}.\]
    In the loose composite \(\mathcal{H} \diamond \mathcal{H}\), the neighbourhood of \(a\) is again a singleton:
    \[\mathcal{N}_{\mathcal{H} \diamond \mathcal{H}}(a) = \{\{\text{grandparents of \(a\)}\}\}.\]
    On the other hand, if \(a\) has two parents---a mother and a father, say---then in the tight composite \(\mathcal{H} \blackdiamond \mathcal{H}\) the neighbourhood of \(a\) has two elements:
    \[\mathcal{N}_{\mathcal{H} \blackdiamond \mathcal{H}}(a) = \{\{\text{maternal grandparents of \(a\)}\}, \{\text{paternal grandparents of \(a\)}\}\}.\]
    Applying the loose composite repeatedly to the F-hypergraph \(\mathcal{H}\) generates the grandparent relation, the great-grandparent relation, and so forth. Using the tight composite instead, one additionally keeps track of branches of the family.
\end{example}

Both \(\blackdiamond\) and \(\diamond\) are binary operations on the set of all F-hypergraph structures on \(A\). In \Cref{subsec:hypergraph-functoriality} we will see that these are the two operations arising from the two semimonad structures on the functor \(\power\power\), as discussed in \Cref{eg:double-powerset-roles}. In particular, that tells us both operations are associative (which can also be checked directly). Thus, the set of F-hypergraph structures on \(A\) acquires two semigroup structures; we will denote these semigroups by \(\mathbb{H}_\blackdiamond(A)\) and \(\mathbb{H}_\diamond(A)\).

\begin{definition}\label{def:hypergraph-roles}
    Let \(\mathcal{H} = (A, (H_i)_{i=1}^k)\) be a multirelational F-hypergraph.
    \begin{itemize}
        \item The \define{semigroup of tight roles} in \(\mathcal{H}\) is the sub-semigroup of \(\mathbb{H}_\blackdiamond(A)\) generated by \(\{H_i\}_{i=1}^k\). That is,
        \[\Role_\blackdiamond (\mathcal{H}) = \langle H_1, \ldots, H_k \rangle_\blackdiamond.\]
        A \define{reduction of tight roles} for \(\mathcal{H}\) is a surjective semigroup homomorphism
        \(\Role_\blackdiamond(\mathcal{H}) \twoheadrightarrow \mathbb{S}\).

        \item The \define{semigroup of loose roles} in \(\mathcal{H}\) is the sub-semigroup of \(\mathbb{H}_\diamond(A)\) generated by \(\{H_i\}_{i=1}^k\). That is,
        \[\Role_\diamond (\mathcal{H}) = \langle H_1, \ldots, H_k \rangle_\diamond.\]
        A \define{reduction of loose roles} for \(\mathcal{H}\) is a surjective semigroup homomorphism
        \(\Role_\diamond(\mathcal{H}) \twoheadrightarrow \mathbb{S}\).
    \end{itemize}
\end{definition}

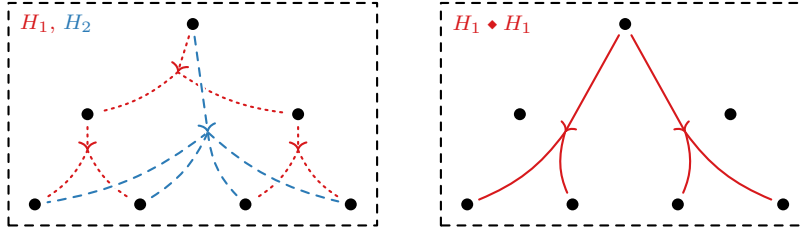
\begin{figure}
    \centering
    \begin{tikzpicture}[line cap=round, line join=round, thick, yscale=0.8]
    \tikzset{ dot/.style={circle, fill, inner sep=1.6pt}, 
    over/.style={preaction={draw=white, line width=\pgflinewidth+.8pt}},}
    \def\bend{16}
    \def\gap{2.2pt} 

    \node[dot] (r) at ( 0, 2.2) {}; 
    \node[dot] (a) at (-1.4, 0.7) {}; 
    \node[dot] (b) at ( 1.4, 0.7) {}; 
    \node[dot] (a1) at (-2.1,-0.8) {}; 
    \node[dot] (a2) at (-0.7,-0.8) {}; 
    \node[dot] (b1) at ( 0.7,-0.8) {}; 
    \node[dot] (b2) at ( 2.1,-0.8) {};

    \coordinate (sr) at ( -.2, 1.4); 
    \coordinate (sa) at (-1.4, 0.1); 
    \coordinate (sb) at ( 1.4, 0.1);

    \draw[->, shorten <=\gap, red, dotted] (r) -- (sr); 
    \draw[->, shorten <=\gap, red, dotted] (a) -- (sa); 
    \draw[->, shorten <=\gap, red, dotted] (b) -- (sb);

    \draw[shorten >=\gap, red, dotted] (sr) to[bend left=\bend] (a); \draw[shorten >=\gap, red, dotted] (sr) to[bend right=\bend] (b);

    \draw[shorten >=\gap, red, dotted] (sa) to[bend left=\bend] (a1);
    \draw[shorten >=\gap, red, dotted] (sa) to[bend right=\bend] (a2);

    \draw[shorten >=\gap, red, dotted] (sb) to[bend left=\bend] (b1);
    \draw[shorten >=\gap, red, dotted] (sb) to[bend right=\bend] (b2);
    
    \coordinate (s) at ( 0.2, 0.4);

    \draw[shorten >=\gap, blue, dashed, over] (s) to[bend left=\bend] (a1); 
    \draw[shorten >=\gap, blue, dashed, over] (s) to[bend right=\bend] (b2); 
    
    \draw[shorten >=\gap, blue, dashed] (s) to[bend left=\bend] (a2);
    \draw[shorten >=\gap, blue, dashed] (s) to[bend right=\bend] (b1); 

    \draw[->, shorten <=\gap, blue, dashed, over] (r) -- (s);

    \draw[dashed] (-2.45,-1.15) rectangle (2.45,2.55); 
    \node[anchor=east] (g) at ( -1.2, 2.2) {\footnotesize {\color{red}{$H_1$}, \color{blue}{$H_2$}}}; 
    
    \end{tikzpicture}
    \qquad
    \begin{tikzpicture}[line cap=round, line join=round, thick, yscale=0.8]
    \tikzset{ dot/.style={circle, fill, inner sep=1.6pt}, } 
    \def\bend{16}
    \def\gap{2.2pt} 

    \node[dot] (r) at ( 0, 2.2) {}; 
    \node[dot] (a) at (-1.4, 0.7) {}; 
    \node[dot] (b) at ( 1.4, 0.7) {}; 
    \node[dot] (a1) at (-2.1,-0.8) {}; 
    \node[dot] (a2) at (-0.7,-0.8) {}; 
    \node[dot] (b1) at ( 0.7,-0.8) {}; 
    \node[dot] (b2) at ( 2.1,-0.8) {};

    \coordinate (sa) at (-.8, 0.4); 
    \coordinate (sb) at ( .8, 0.4);

    \draw[->, shorten <=\gap, red] (r) -- (sa); 
    \draw[->, shorten <=\gap, red] (r) -- (sb);

    \draw[shorten >=\gap, red] (sa) to[bend left=\bend] (a1);
    \draw[shorten >=\gap, red] (sa) to[bend right=\bend] (a2);

    \draw[shorten >=\gap, red] (sb) to[bend left=\bend] (b1);
    \draw[shorten >=\gap, red] (sb) to[bend right=\bend] (b2);
    
    \draw[dashed] (-2.45,-1.15) rectangle (2.45,2.55); 
    
    \node[anchor=east] (g) at ( -1.1, 2.2) {\footnotesize {\color{red}{$H_1 \blackdiamond H_1$}}}; 
    
    \end{tikzpicture}

    \caption{For the multirelational F-hypergraph on the left, \(\mathcal{H} = (A,(H_1,H_2))\), the semigroups of tight and loose roles differ. See \Cref{eg:tight-loose-different}.}
    \label{fig:tight-loose-different}
\end{figure}

The following example confirms that the semigroups of tight and loose roles for a given \(\mathcal{H}\) can differ.

\begin{example}
\label{eg:tight-loose-different}
    Consider the multirelational F-hypergraph \(\mathcal{H}\) in \Cref{fig:tight-loose-different}. Notice that \(H_1 \diamond H_1 = H_2 \neq H_1 \blackdiamond H_1\). It follows that \(\Role_\blackdiamond(\mathcal{H})\) has three non-zero elements, \(H_1\), \(H_2\) and \(H_1 \blackdiamond H_1\), while \(\Role_\diamond(\mathcal{H})\) has only two, \(H_1\) and \(H_2\).
\end{example}

\begin{remark}
    The loose composition operation has been introduced independently in modal logic. Van Benthem, Bezhanishvili and Enqvist~\cite{vanBenthem2019} develop a dynamic extension of \emph{instantial neighbourhood logic}, interpreted over \emph{neighbourhood frames}: structures consisting of {states} and specifying which sets of possible {outcomes} are achievable at each state. Formally, neighbourhood frames are exactly \(\power\power\)-coalgebras. The dynamic logic reasons about the effects of applying \emph{programs} transforming these neighbourhood frames. To do this, one needs a semantics for the effect of applying one program after another. This semantics precisely corresponds to our loose composition. As they point out, this composition operation is \emph{bisimulation-safe}, meaning that it preserves the bisimilarity of neighbourhood frames. Our \Cref{thm:functoriality-hypergraphs} likewise ties loose composition to bisimulation (via positional reductions), and also shows that it is not an \emph{ad hoc} choice: it is the composition induced by one of the semimonad structures on $\power\power$.
\end{remark}


\subsection{Compatibility of positional and role analysis}
\label{subsec:hypergraph-functoriality}

We will close by applying our main result---\Cref{thm:functoriality-coalgebras}---to prove that both forms of role analysis introduced in \Cref{subsec:hypergraph-role-analysis} are compatible with positional analysis for F-hypergraphs, as defined in \Cref{subsec:hypergraph-positional-analysis}. Just as in the case of graphs, we will express this compatibility in terms of a functoriality theorem.

The theorem says, in particular, that every positional reduction of F-hypergraphs \(\mathcal{H} \twoheadrightarrow \mathcal{K}\) induces a role reduction \(\Role_\blackdiamond(\mathcal{H}) \twoheadrightarrow \Role_\blackdiamond(\mathcal{K})\) \emph{and} a role reduction \(\Role_\diamond(\mathcal{H}) \twoheadrightarrow \Role_\diamond(\mathcal{K})\); this is the analogue, for hypergraphs, of the classical statement \Cref{prop:induced-map-graph-case} for graphs. It also says a little more: that if one performs a sequence of several positional reductions, then the induced role reductions are guaranteed to behave consistently.

To state the theorem, let \(k\FHGraph_\mathsf{PR}\) denote the category of \(k\)-relational F-hypergraphs and positional reductions. As before, \(\mathbf{SGrp}_\mathsf{Q}\) denotes the category of semigroups and surjective semigroup homomorphisms.

\begin{theorem}
\label{thm:functoriality-hypergraphs}
    The semigroups of tight and loose roles are both functorial with respect to positional reductions. That is, they define functors
    \begin{equation*}
    \label{eq:tight-functorial}
        \Role_\blackdiamond \colon k\FHGraph_\mathsf{PR} \to \mathbf{SGrp}_\mathsf{Q}
    \end{equation*}
    and
    \begin{equation*}
    \label{eq:loose-functorial}
        \Role_\diamond \colon k\FHGraph_\mathsf{PR} \to \mathbf{SGrp}_\mathsf{Q}.
    \end{equation*}
\end{theorem}

\begin{proof}
    In \Cref{eg:hypergraph-coalgebras} we saw that the neighbourhood coalgebra \(\mathcal{N}_{(-)}\) defines a bijection between F-hypergraphs and \(\power\power\)-coalgebras. By \Cref{lem:pos-red-equiv}, a positional reduction of F-hypergraphs \(\mathcal{H} \twoheadrightarrow \mathcal{K}\) is precisely a positional reduction of \(\power\power\)-coalgebras \(\mathcal{N}_\mathcal{H} \twoheadrightarrow \mathcal{N}_\mathcal{K}\). Thus, the neighbourhood coalgebra determines an equivalence of categories
    \[\mathcal{N}_{(-)} \colon k\FHGraph_\mathsf{PR} \xto{\simeq} k\Coalg(\power\power)_\mathsf{PR}.\]
    Meanwhile, by \Cref{thm:functoriality-coalgebras}, the two semigroups of roles \(\Role_{\mathbb{PP}_1}\) and \(\Role_{\mathbb{PP}_2}\) (introduced in \Cref{eg:double-powerset-roles}) each define a functor \(k\Coalg(\power\power)_\mathsf{PR} \to \mathbf{SGrp}_\mathsf{Q}\).
    Our claim is that we have \(\Role_{\blackdiamond}(\mathcal{H}) \cong \Role_{\mathbb{PP}_1} (\mathcal{N}_{\mathcal{H}})\) and \(\Role_{\diamond}(\mathcal{H}) \cong \Role_{\mathbb{PP}_2}(\mathcal{N}_{\mathcal{H}})\) for every \(\mathcal{H}\) in \(k\FHGraph_\mathsf{PR}\), and hence a diagram of functors
    \[
    \begin{tikzcd}
        k\FHGraph_\mathsf{PR} \arrow{r}{\mathcal{N}_{(-)}} \arrow[bend left=30]{rr}{\Role_\blackdiamond} \arrow[bend right=32, swap]{rr}{\Role_\diamond} & k\Coalg(\power\power)_\mathsf{PR} \arrow[shift left=1]{r}{\Role_{\mathbb{PP}_1}} \arrow[shift right=.9, swap]{r}{\Role_{\mathbb{PP}_2}} & \mathbf{SGrp}_\mathsf{Q}.
    \end{tikzcd}
    \]
    All we need to prove is that for every set \(A\) and every pair of F-hypergraphs \((A,H)\) and \((A,K)\), we have \(\mathcal{N}_{K} \ast_1 \mathcal{N}_{H} = \mathcal{N}_{K \blackdiamond H}\) and \(\mathcal{N}_{K} \ast_2 \mathcal{N}_{H} = \mathcal{N}_{K \diamond H}\).
    
    We consider the tight composite first. By definition of composition in the semicategory \(\Set_{\mathbb{PP}_1}\), we have
    \[
    \begin{tikzcd}
        A \arrow[swap]{r}{\mathcal{N}_{H}} \arrow[bend left=10]{rrrr}{\mathcal{N}_{K} \ast_1 \mathcal{N}_{H}} &
        \power\power(A) \arrow[swap]{r}{\power\power(\mathcal{N}_{K})} &
        \power\power\power\power(A) \arrow[swap]{r}{\mu_{\power\power(A)}} &
        \power\power\power(A) \arrow[swap]{r}{\mu_{\power(A)}} &
        \power\power(A)
    \end{tikzcd}
    \]
    where \(\mu\) is the multiplication of the powerset monad. Taking an element \(a \in A\) and tracking it through these four maps, we see that
    \begin{align*}
        a \mapsto \mathcal{N}_{H}(a) &\mapsto \left\{ \{\mathcal{N}_{K}(b) \mid b \in V \} \mid V \in \mathcal{N}_{H}(a) \right\} \\
        &\mapsto \bigcup_{V \in \mathcal{N}_{H}(a)} \{\mathcal{N}_{K}(b) \mid b \in V\} \\
        &\mapsto \bigcup_{V \in \mathcal{N}_{H}(a)} \left( \bigcup_{b \in V} \mathcal{N}_{K}(b)\right).
    \end{align*}
    That is, for each \(a \in A\) we have
    \[(\mathcal{N}_{K} \ast_1 \mathcal{N}_{H})(a) = \{U \mid \exists \: (a,V) \in H \text{ and } b \in V \text{ such that } (b,U) \in K\},\]
    which is exactly \(\mathcal{N}_{K \blackdiamond H}(a)\).
    
    Now we consider the loose composite. By definition of composition in the semicategory \(\Set_{\mathbb{PP}_2}\), we have
    \[
    \begin{tikzcd}
        A \arrow[swap]{r}{\mathcal{N}_{H}} \arrow[bend left=10]{rrrr}{\mathcal{N}_{K} \ast_2 \mathcal{N}_{H}} &
        \power\power(A) \arrow[swap]{r}{\power\power(\mathcal{N}_{K})} &
        \power\power\power\power(A) \arrow[swap]{r}{\power\power(\mu_A)} &
        \power\power\power(A) \arrow[swap]{r}{\power(\mu_A)} &
        \power\power(A).
    \end{tikzcd}
    \]
    Taking \(a \in A\) and tracking it through these four maps, we see that
    \begin{align*}
        a \mapsto \mathcal{N}_{H}(a) &\mapsto \left\{ \{\mathcal{N}_{K}(b) \mid b \in V \} \mid V \in \mathcal{N}_{H}(a) \right\} \\
        &\mapsto \left\{ \left\{ \bigcup_{U \in \mathcal{N}_{K}(b)} U \: \middle| \: b \in V \right\} \: \middle| \: V \in \mathcal{N}_{H}(a) \right\} \\
        &\mapsto \left\{ \bigcup_{b \in V} \left(\bigcup_{U \in \mathcal{N}_{K}(b)} U \right) \: \middle| \: V \in \mathcal{N}_{H}(a) \right\}.
    \end{align*}
    That is, for each \(a \in A\) we have
    \[(\mathcal{N}_{K} \ast_2 \mathcal{N}_{H})(a) = \left\{ W \mid \exists \: (a,V) \in H \text{ such that } W = \bigcup_{b \in V} \bigcup_{(b,U) \in K} U \right\},\]
    which is exactly \(\mathcal{N}_{K \diamond H}(a)\).
\end{proof}


\section{Conclusions and future work}\label{sec:future_work}

In this paper we have introduced a coalgebraic framework to generalize role and positional analysis from graphs to hypergraphs, tying the two types of analysis together via a functoriality result. We believe this to be the first step in a promising line of investigation, with a number of natural extensions and applications, some of which we indicate briefly below.

\paragraph{Arbitrary directed hypergraphs}
To keep the framework conceptually simple, we restricted to F-hypergraphs: directed hypergraphs in which the source of each hyperedge is a singleton set. Our framework and functoriality result can, however, be extended to arbitrary directed hypergraphs,  at the cost of moving from coalgebras to \emph{dialgebras}, which are a generalisation of both algebras and coalgebras \cite{POLL2001289}.

\paragraph{Approximate notions of equivalence}
In real-world applications of positional analysis, researchers often work not with strict notions of equivalence, but rather with measures of approximate equivalences, quantifying the extent to which two actors are equivalent; such measures include cosine similarity \cite[Chapter 7.12.1]{newman10} and measures of approximate regular equivalences \cite[Chapter 12.4.4]{Wasserman_Faust_1994}. To extend our framework and results to this setting, there are several potential directions to explore, such as coalgebraic behavioural metrics \cite{baldan14} which allow to encode distances between actors in a system,  enriched generalisations of coalgebras \cite{BKV11}, or approximate bisimulations in the context of automata \cite{9268079}.

\paragraph{Non-associative compositions}
While our framework yields associative compositions of relations, in some scenarios the meaning of a composite relation may depend on context, so that different bracketings represent genuinely different composites. An interesting question is whether functoriality can still be obtained in such a setting, perhaps by working with \emph{lax} coalgebra morphisms and replacing semigroups with nonunital magmas. It is unclear what bisimulations and regular equivalences would look like in such a setting. We stress that these questions are also relevant to graphs: for instance, when modelling kinship networks in which the meaning of relations is gender-dependent (see Remark \ref{R:non-associative}).

\paragraph{Real-world applications}
To deploy role and positional analysis in real-world applications, there are several directions in which our work could be taken. First, while in applications the semigroups of roles we consider will always be finite, in practice even reductions of these semigroups may be too large to easily interpret. In particular, it might be difficult to assign a clear meaning to a compound relation given by a long string of relations. One way to address this is by working with truncations of the semigroups, as discussed in \cite[Section 9.2]{OtterPorter2020}. Second, whilst applications often focus solely on computing specific bisimulations, such as the maximal bisimulation, we believe that the full lattice of bisimulations provides a wealth of information about a given social system that ought to be harnessed.

\section*{Acknowledgements}
Our collaboration  started during the 2022 Mathematics Research Community (MRC) on Applied Category Theory \cite{act}. We are profoundly indebted to several MRC participants, with whom we discussed many related ideas and explored potential research directions which ultimately were not pursued. In particular, we would like to thank: Daniel Cicala, Zachary Flores, Abigail Hickok,  Elise McMahon, Nikola Mili\'{c}evi\'{c},  Rachel Hardeman Morrill, Joseph Randich, Jordan Sawdy, Joshua Tan, and Angelo Taranto. 

We also thank Wolfgang Poiger and Helle Hvid Hansen for bringing instantial neighbourhood logic to our attention at CALCO 2025.
    
We thank the AMS for the generous support during and after the MRC, through the National Science Foundation under Grant Number DMS 1641020.


\printbibliography

\end{document}